\documentclass[10 pt]{article}
\author{Ali Tajer$^*$ \and H. Vincent Poor$^\dag$}
\date{}

\usepackage{graphicx, subfigure}
\usepackage{amsmath, amssymb}
\usepackage{dsfont}
\usepackage{color}

\newcommand{\med}{\;|\;}

\newtheorem{theorem}{Theorem}
\newtheorem{corollary}{Corollary}
\newtheorem{lemma}{Lemma}

\newtheorem{definition}{Definition}

\def \D{{\cal D}}

\def \L{{\cal L}}

\def \X{{\cal X}}
\def \F{{\cal F}}

\def \U{{\cal U}}
\def \sP{{\cal P}}

\def \H{{\sf H}}
\def \T{{\sf T}}
\def \D{{\sf D}}
\def \PP{{\sf P}}

\def \P{\mathbb{P}}

\newcommand{\aeq}{\overset{\rm a}{=}}

\newcommand{\dff}{\stackrel{\scriptscriptstyle\triangle}{=}}
\usepackage{accents}
\newlength{\dhatheight}
\newcommand{\doublehat}[1]{%
\settoheight{\dhatheight}{\ensuremath{\hat{#1}}}%
\addtolength{\dhatheight}{-0.35ex}%
\hat{\vphantom{\rule{1pt}{\dhatheight}}%
\smash{\hat{#1}}}}

\title{\huge Quick Search for Rare Events}

\begin{document}
\maketitle
\allowdisplaybreaks

\begin{center}
{\em First version: November 16, 2011\\
Current version: September 30, 2012 }
\end{center}
\begin{abstract}
Rare events can potentially occur in many applications. When manifested as opportunities to be exploited, risks to be ameliorated, or certain features to be extracted, such events become of paramount significance. Due to their sporadic nature, the information-bearing signals associated with rare events often lie in a large set of irrelevant signals and are not easily accessible.  This paper provides a statistical framework for detecting such events so that an optimal balance between detection {\em reliability} and  {\em agility}, as two opposing performance measures, is established. The core component of this framework is a sampling procedure that adaptively and quickly focuses the information-gathering resources on the segments of the dataset that bear the information pertinent to the rare events. Particular focus is placed on Gaussian signals with the aim of detecting signals with rare mean and variance values. {\renewcommand{\thefootnote}{}\footnotetext{$^*$Electrical \& Computer Engineering Department, Wayne State University, Detroit, MI 48202.
email: {\tt tajer@wayne.edu}
}}
{\renewcommand{\thefootnote}{}\footnotetext{
$^\dag$Electrical Engineering Department, Princeton University, Princeton, NJ 08544. e-mail: {\tt poor@princeton.edu.}
}}{\renewcommand{\thefootnote}{}\footnotetext{This research was supported in part by the National Science Foundation under Grant DMS-11-18605.}}

\end{abstract}

\section{Introduction}
\label{sec:general}

The problem of searching for scarce and at the same time significant events with certain statistical behavior in observed information streams is a classical one and has attracted attention in a wide variety of fields over the past few decades. Such events can broadly model  at least three categories of problems. One concentrates on seeking opportunities in a range of domains including trading in finance \cite{finance1, finance2, finance3} and spectrum sensing in telecommunication \cite{telecom1, telecom2}. The second category pertains to minimizing risks and avoiding catastrophes in applications such as risk analysis in econometrics \cite{econ1, econ2, econ3, econ4, econ5, econ6}, blackout cascade avoidance in energy systems \cite{energy1, energy2}, intrusion identification in network security \cite{sec1, sec2, sec3, sec4, sec5, sec6}, fraud detection \cite{fraud1, fraud2}, and seismology. And the third category deals with extracting certain data segments with pre-specified probabilistic features with applications in image and video analysis, medical diagnosis, neuroscience, and remote sensing (seismic, sonar, radar, biomedical)~\cite{sensing1, sensing2, sensing3}

While detecting the rare events in some of these applications is time-insensitive, in some other applications time is of the essence and it is important to devise timely and reliable decision-making mechanisms. Such time-sensitivity can be due to the transient nature of the opportunities that are attractive only when detected quickly, or due to the substantial costs that risks can incur if not detected and managed swiftly, or for allowing for real-time processing of the information.

All these applications can be modeled in terms of collection of events that can be broadly categorized into two groups. One group constitutes the majority of the events, which are deemed normal and occur most of the time, and the other group consists of  the events that occur rarely but are of extreme significance to the observer. In this paper, without loss of generality, we assume that the group of all normal events share the same statistical behavior and the group of rare events also share identical statistical behavior, albeit different from that of the normal events. This dichotomous model is mainly to focus the attention on the discrepancy between the rare and the normal events and can be easily generalized to models that involve multiple statistical behaviors for each group. Hence, we assume that the observer has access to a collection of $n$ sequences of random observations $\X^1,\dots,\X^n$, each modeling one event. Each sequence $\X^i$ consists of independent and identically distributed (i.i.d.) measurements $\X^i\dff\{X^i_1,X^i_2,\dots\}$ taking values in the set $\Omega=\mathbb{R}$ endowed with a $\sigma$-field $\F$ of events, obeying one of the two hypotheses.
\begin{equation}\label{eq:H}
    \begin{array}{cc}
      \H_0: & X^i_j\sim F_0, \quad j=1,2,\dots\\
      \H_1: & X^i_j\sim F_1, \quad j=1,2,\dots
    \end{array}
\end{equation}
where $F_0$ and $F_1$ denote the cumulative distribution functions (cdfs) of two distinct distributions on $(\Omega,\mathcal{F})$. The distribution $F_0$ models the statistical behavior of the normal events and the distribution $F_1$ models the statistical behavior of the rare events. For convenience, we assume that $F_0$ and $F_1$ have probability density functions (pdfs) $f_0$ and $f_1$, respectively. Each sequence $\X^i$ is generated  by $F_0$ or $F_1$ independently of the rest of the sequences and we assume that hypothesis $\H_1$ (a rare event) occurs with prior probability $\epsilon_n$. In order to incorporate the rareness of the sequences generated by $F_1$, we assume that $\epsilon_n=o(1)$.

The goal of quick search is to identify {\em one or more} rare events among all $n$ given events through 1) designing an information-gathering process for collecting information from the sequences $\X^1,\dots,\X^n$, and 2) delineating optimal decision rules. Designing a quick search process involves a tension between two performance measures, one being the aggregate amount of information accumulated (i.e., the number of observations made) and the other being the reliability (or cost) of the decision. In this paper we design an optimal information-gathering process that maximizes the decision reliability subject to a {\em hard}\footnote{By a hard constraint we mean that the {\em aggregate} number of observations made cannot exceed a specified level.} constraint on the aggregate number of observations we are allowed to make from $\X^1,\dots,\X^n$.

\subsection{Related Literature}

The quick search problem is closely related to the sequential detection literature, with two major differences. First, the existing approaches in sequential detection that can be applied to the quick search problem at hand often optimize a balance  between decision reliability and a {\em soft}\footnote{By soft constraint we mean that the {\em expected} number of the observations made cannot exceed a specified level.} constraint on the number of observations. This is in contrast with the quick search setting that enforces a {\em hard} constraint on the sensing resources. Secondly, sequential detection often aims to identify  {\em all} rare events or {\em only one} rare event, whereas our setting offers the flexibility to identify one or more rare events. The most relevant sequential detection solutions for identifying all and only one rare event (generated by $F_0$) are the sequential probability ratio test (SPRT)~\cite{Wald:AMS49} and the cumulative sum (CUSUM) test~\cite{Lai:IT11}, respectively. More specifically, by denoting the true hypothesis and a decision about sequence $\X^i$ by $\T_i\in\{\H_0,\H_1\}$ and $\D_i\in\{\H_0,\H_1\}$, respectively, the Type-I and Type-II detection error probabilities corresponding to sequence $\X^i$ are
\begin{equation}\label{eq:det_error}
    \PP_i^1=\P(\D_i=\H_1\med \T_i=\H_0)\quad\mbox{and}\quad \PP_i^2=\P(\D_i=\H_0\med \T_i=\H_1)\ .
\end{equation}
When the objective is to identify {\em all} rare events (generated by $F_1$), minimizing the {\em average} number of observations made with constraints on $\PP_i^1$ and $\PP^i_2$ can be decomposed into minimizing the {\em average} number of observations necessary for deciding between $\H_0$ and $\H_1$ for each individual sequence with the same reliability constraints~\cite{Wald:AMS45}. The optimal test for each, which is the test that requires the smallest number of observations and satisfies the reliability constraints, is the SPRT~\cite{Wald:AMS49}. On the other hand, when the objective is to identify only one rare event (generated by $F_1$) there exist, broadly, two classes of solutions. In one class $n$ is finite and it is known {\em a priori} that there exists only one rare event (this is called scanning problem). The optimal detector that identifies the rare event for the Brownian motion setting is given in~\cite{Zigangirov:TPA66} and \cite{Dragalin:Metrica96}. In the other class $n=\infty$ so that it is ensured that almost surely there exists a rare event. For this class the optimal test, which is the test that minimizes the {\em average} delay subject to a constraint on Type-II error probability, is the CUSUM test~\cite{Lai:IT11}.

We also remark briefly on the connection between the quick search problem and the problem of sparse support recovery in noisy environments. These two problems become equivalent (in settings and not objectives) when the events space is mapped to the sparse signal space, the rare events are mapped to the constituents of the support of the sparse signal, and the stochastic observations of the normal and rare events follow the distributions of noisy observations of the normal and sparse components of the sparse signal. A recent line of research in sparse recovery relevant to the quick search problem is the notion of adaptive sampling (sensing) in high-dimensional noisy data with sparsity structures, in which the objective is to estimate the support of the signal~\cite{Haupt:AISTATS09, Haupt:IT11}. The works in ~\cite{Haupt:AISTATS09}, \cite{Haupt:IT11}, and \cite{Malloy:ISIT11}  propose adaptive sampling procedures that effectively focus the sensing resources on the segments of the data that have higher likelihood of containing the support. Despite the similarities in the settings, however, there is one major discrepancy in the objectives of sparse recovery and quick search, and that is the equivalent of the goal of quick search in the sparse recovery setting becomes identifying {\em one subset of the support} of desired length, whereas the goal of sparse recovery is to identify the entire support of the signal.  It is noteworthy that such notion of {\em fractional} support recovery is also studied in the context of compressive sensing through  {\em non-adaptive} data acquisition~\cite{Reeves:ISIT08}. The research in~\cite{Reeves:ISIT08} considers a sparse signal with known support size (unlike our setting that the number of rare events is stochastic) and takes one set of low-dimensional observation and analyzes the interplay among the sampling rate, the fraction of the support to be recovered, and the necessary and sufficient conditions on signal power for recovering the support reliably.


\subsection{Contributions}

The objective of quick search in this paper is to identify a {\em fraction} of the rare events (as opposed to the existing literature which targets at identifying either all or only one rare events). This less-investigated objective  aims to fill the gap between the two well-studied extreme (all or one) objectives. This shift of objective allows the detector to tolerate a higher level of Type-I errors, i.e., the detector can afford to miss some of the rare events, in favor of quickly detecting the fraction of interest of the rare events (i.e. sequences generated by $F_1$). In general the number of rare events is a random variable between 0 and $n$ and is unknown {\em a priori}.

Besides the objective, another major distinction from the existing literature is the constraint enforced on the sampling resources. This paper considers the less-investigated scenario in which there exists a {\em hard} constraint on the sampling resources. This is in contrast with the existing literature which often incorporates the statistical {\em average} of the sampling budget in the analysis. We remark that to the best of our knowledge, the general problem of designing the optimal sampling strategy, which distributes a limited sampling budget among the sequences under observation in order to yield the highest level of reliability in detectng one, a fraction, or all of the sequences generated by $F_1$ is an open problem. In this paper we focus on the asymptotic performance when the number of sequences $n$ is sufficiently large and provide the asymptotically optimal sampling procedure when the distributions $F_0$ and $F_1$ are Gaussian with either different mean or different variance values.

To solve the quick search problem with the aforementioned objective and constraint, we design a sequential and data-adaptive information-gathering procedure and detection rule.

\subsection{Summary of Results}
An adaptive and sequential sampling process is designed that dynamically makes one of the following three data-driven decisions at each time:  1) it lacks sufficient evidence to identify the rare events and decides to postpone the decision and collects more data (observation), or 2) reaches some confidence to eliminate {\em roughly} $\alpha\in(0,1)$ fraction of the events that are deemed to be the weakest candidates for being rare events, but yet does not have enough evidence to make the final decision (refinement), or 3) it has accumulated enough information to identify the rare events of interest (detection). As one major result of this paper we characterize the the {\em asymptotically} optimal (in the asymptote of large number of events $n$) allocation of the sampling resources among the sequences and the pertinent sampling procedure, which consists of consecutive rounds of {\em coarse} observations and refinement actions, followed by consecutive cycles of {\em fine} observations. The number of each these cycles (refinement and observation) is also determined as a function of the number of available sampling budget.

It is confirmed in different contexts that in order to recover the rare events (or the support of the sparse signal) reliably, the power of the observed data should be scaling with the data size $n$ (cf. \cite{Jin:2006_1}, \cite{Reeves:ISIT08}, \cite{Haupt:AISTATS09}, \cite{Haupt:IT11}, and~\cite{Malloy:ISIT11}). We apply the proposed adaptive sampling procedure for setting that the normal and rare events are generated according to Gaussian distributions with either different mean or different variance values and characterize necessary and conditions on scaling rates of the mean and variance values for identifying the desired number of rare events. These scaling rates of the means and variances are functions of the data dimension $n$, the available amount of sampling budget, and the frequency of the rare events. More specifically, if we denote the likelihood of an event being a rare event by $\epsilon_n$ and define ${\varepsilon_n}\in(0,1)$ as
\begin{equation}
{\varepsilon_n}\dff\frac{\ln n\epsilon_n}{\ln n}\ ,
\end{equation}
for successfully identifying a  {\em small fraction} of the rare events when the normal events are distributed as ${\cal N}(\mu_0,1)$ and the rare events as ${\cal N}(\mu_1,1)$ via the adaptive procedure a necessary and sufficient condition is that in the asymptote of large $n$
\begin{equation}\label{eq:muuu}
\frac{(\mu_0-\mu_1)^2}{\log n}\;>\;\frac{(1-\sqrt{{\varepsilon_n}})^2}{c}\ ,
\end{equation}
where $c$ is a constant determined primarily by the available sampling budget and the number of refinement cycles. We also show that when the normal and rare events are distributed as ${\cal N}(0,A_0)$ and ${\cal N}(0,A_1)$, respectively, the counterpart necessary and sufficient condition is
\begin{equation}\label{eq:AAA}
\frac{\ln A_0/A_1}{\ln n}\;>\; \frac{2(1-{\varepsilon_n})}{c}
\end{equation}
Finally, in order to assess the gains of the data-adaptive sampling process in comparison with a non-adaptive procedure we assess the scaling laws of the mean and variance values. It is shown when the adaptive and non-adaptive procedures enjoy the same sampling budget, for the mean and variance values we identify the same behavior is in \eqref{eq:muuu} and \eqref{eq:AAA}, respectively, with the exception that the constant $c$ is decreased to $c\cdot\alpha^{K}$, where $\alpha\in(0,1)$ was defined earlier and $K$ is related to the number of refinement cycles. In another interpretation of the results, when the mean and variance have identical scaling behavior in both adaptive and non-adaptive sampling settings, for achieving identical level of reliability in detecting the rare events, the adaptive procedures requires {\em roughly} on the order of $\alpha^{-K}$ less sampling resources.
\section{Problem Statement}
\label{sec:problem}

\subsection{Sampling Model}
\label{sec:sampling}
With the ultimate objective of identifying $T_n=o(n\epsilon_n)$ (i.e., a small fraction of the rare events) the proposed sampling procedure is initiated by making observations from {\em all} sequences $\X^1,\dots,\X^n$. Based on these rough observations a fraction of the sequences that are least-likely generated by $F_0$ are discarded and the rest are retained for further and more accurate scrutiny. Repeating this procedure successively refines the search support and progressively focuses the observations on the more promising sequences. More specifically, at each time the sampling procedure selects a subset of the sequences $\X^1,\dots,\X^n$ and takes one sample from each of these sequences. Upon collecting these samples, it takes one of the following actions:
\begin{enumerate}
  \item [$\rm A_1$] {\bf (Detection):} stops further sampling and identifies $T_n$ sequences that have the highest likelihood of being rare events (generated by $F_1$);
  \item [$\rm A_2$] {\bf (Observation):} continues to further observe the same set of sequences in order to gather more information about their statistical behavior; or
  \item [$\rm A_3$] {\bf (Refinement):} discards a portion of the sequences and declares that they are most likely normal events (generated by $F_0$). When a sequence is discarded it will be deemed a weak candidate for being a rare event (generated by $F_1$) and will not be observed anymore, while the remaining sequences are retained for more scrutiny. By denoting the number of sequences retained prior to a refinement action by $\ell$, the number of sequences that this action discards is  $(1-\alpha)(\ell-T_n)$ for some $\alpha\in (0,1)$. Discarding the sequences at this rate ensures that at least $T_n$ sequences will be retained for the final detection action (action ${\rm A}_1)$.
\end{enumerate}
We denote the set of the indices of the sequences observed at time $t\in\mathbb{N}$ by $\L_t$. As we initialize the information-gathering procedure by including all sequences for observation we have $\L_1=\{1,\dots,n\}$. Also, we denote the stopping time of the procedure, i.e., the time after which detection (action $\rm A_1$) is performed, by $\tau$. Furthermore, we define the switching function $\psi:\{1,\dots,\tau\}\rightarrow\{0,1\}$ to model actions $\rm A_2$ (observation) and $\rm A_3$ (refinement). At each time $1\leq t\leq\tau-1$ we set $\psi(t)=0$ if we decide in favor of performing observation, while $\psi(t)=1$ indicates a decision in favor of performing refinement, i.e.,
\begin{equation}\label{eq:psi}
\forall t\in\{1,\dots,\tau-1\}:\quad \psi(t)=\left\{
\begin{array}{ll}
0 & \mbox{action }{\rm A}_2\;\;\mbox{and}\;\; \L_{t+1}=\L_t\\
1 & \mbox{action }{\rm A}_3\;\;\mbox{and}\;\; \L_{t+1}\subset\L_t\\
\end{array}\right.\ .
\end{equation}
Let $X_t^i$ denote the observations made from sequence $i\in\L_t$ at time $t$ and denote the $\sigma$-algebra generated by the observation $\{X^{i}_1,\dots,X_t^i\}$ by
\begin{equation}\label{eq:filter}
\forall i\in\L_t:\quad \F_t^i=\sigma(X^{i}_1,\dots,X_t^i)\ .
\end{equation}
Given $\F_t^i$, we denote the posterior probability that the sequence $\X^i$, for $i\in\L_t$, is generated by $F_1$ by $\pi_t^i\dff\P(\T_i=\H_1\med \F_t^i)$. Invoking the independence among the observations $\{X^i_1,\dots,X_t^i\}$ provides
\begin{equation}\label{eq:pi}
\pi_t^i=\left[1+\frac{1-\epsilon_n}{\epsilon_n}\prod_{u=1}^t\frac{f_0(X_u^i)}{f_1(X_u^i)}\right]^{-1}\ .
\end{equation}
By defining the likelihood ratio
\begin{equation}\label{eq:Lambda}
\Lambda_t^i\dff\prod_{u=1}^t\frac{f_0(X_u^i)}{f_1(X_u^i)}\ ,
\end{equation}
we have
\begin{equation}\label{eq:pi2}
\pi_t^i=\left[1+\frac{1-\epsilon_n}{\epsilon_n}\;\Lambda_t^i\right]^{-1}\ ,
\end{equation}
and the actions $\rm A_1,\;A_2$, and $A_3$ can be formalized as follows.
\begin{enumerate}
\item [$\rm A_1$:] At the stopping time $\tau$ identify the set $\U\subseteq \L_{\tau}$ as the detector's decision according to
\begin{eqnarray}\label{eq:U}
\nonumber \U & = & \arg\max_{\U\subseteq\L_{\tau}:\;|\U|=T_n}\P\left(\forall i\in \U:\;\T_i=\H_1\med \{\F^i_{\tau}:\;i\in\L_{\tau}\}\right)\\
\nonumber & = & \arg\max_{\U\subseteq\L_{\tau}:\;|\U|=T_n}\prod_{i\in\U}\P\left(\T_i=\H_1\med \F^i_{\tau}\right)\\
\nonumber & = & \arg\max_{\U\subseteq\L_{\tau}:\;|\U|=T_n}\prod_{i\in\U}\pi^i_{\tau}\\
& = & \mbox{the indices of the $T_n$ smallest elements of the set } \{\Lambda^i_{\tau}:\;i\in\L_{\tau}\}\ .
\end{eqnarray}
\item [$\rm A_2$:] At time $t$ the decision is to further measure the same set of sequences (i.e. $\L_{t+1}=\L_t$), and set
\begin{equation}\label{eq:A2}
\Lambda^i_{t+1}=\Lambda_t^i\cdot\frac{f_0(X^i_{t+1})}{f_1(X^i_{t+1})}\ .
\end{equation}
\item [$\rm A_3$:] At time $t$ the decision is to refine $\L_t$ and the set $\L_{t+1}$ is obtained as
\begin{eqnarray}\label{eq:A3}
\nonumber \L_{t+1} & = & \arg\max_{\L\subseteq\L_t:\;|\L|=\bar\alpha}\P\left(\forall i\in \L:\;\T_i=\H_1\med \{\F^i_{t}:\;i\in\L_{t}\}\right)\\
& = & \mbox{the indices of the $\bar\alpha$ smallest elements of the set } \{\Lambda^i_{t}:\;i\in\L_{t}\}\ ,
\end{eqnarray}
where by recalling the description of the refinement action we can write
\begin{equation}\label{eq:alpha_bar}
    \bar\alpha\;\dff\;\lfloor\alpha(|\L_t|-T_n)\rfloor\;+\;T_n\ .
\end{equation}

\end{enumerate}

\subsection{Optimal Sampling}
\label{sec:optimal}

Characterizing the experimental design and decision rules for performing quick search involves an interplay between two performance measures, one being the aggregate number of observations made and the other being the frequency of erroneous detection. Any improvement in either of these performance measure penalizes the other one and the optimal design of such procedures typically involves optimizing a tradeoff between them. Furthermore, once the optimal aggregate number of observations is determined, we also need to ascertain how these observations should be split and allotted to the refinement and detection actions. For a given stopping time $\tau$ and a given sequence of switching functions $\bar\psi(\tau)\dff\{\psi(1),\psi(2),\dots,\psi(\tau-1)\}$, the probability of erroneous detection, that is the probability that the detected sequences includes a normal event (generated by $F_1$), is
\begin{equation}\label{eq:Pn}
\PP_n(\tau,\bar\psi(\tau))\;\dff\;\P\Big(\big|\{i\in\U:\;\T_i=\H_0\}\big|\neq 0\Big)\ .
\end{equation}
Our goal is to minimize this above detection error probability over all possible stopping times $\tau$, all switching rules $\bar\psi$, and all possible allocations of the observations to refinement and detection actions subject to two {\em hard} constraints. One constraint incorporates the impact of the aggregate number of observations made, and the other one captures the cost of the refinement actions. Specifically, for the first constraint the aggregate number of observations normalized by the number of sequences is constrained to be less than a pre-specified value $S$. As the second constraint, the number of refinement actions must also be smaller than a pre-specified value $K$. This latter constraint captures the cost incurred by each refinement action, which is the permanent loss of the sequences discarded after the refinement actions. This optimization problem can be formalized as
\begin{equation}\label{eq:P}
\sP_n(S,K)\dff\left\{
\begin{array}{ll}
\inf_{\tau,\bar\psi(\tau)} & \PP_n(\tau,\bar\psi(\tau))\vspace{.1 in}\\
{\rm s. t. } & \frac{1}{n}\sum_{t=1}^{\tau-1}|\L_t|\leq S\vspace{.1 in}\\
 & \sum_{t=1}^{\tau}\psi(t)\leq K
\end{array}
\right. \ .
\end{equation}
We solve this problem in the asymptote of large $n$ for two hypothesis testing problems that test the mean and the variance of Gaussian observations, i.e., the problems of form \eqref{eq:H} with the hypothesis $\H_m$, for $m\in\{0,1\}$, are:
\begin{eqnarray}\label{eq:H_Gaussian}
\label{eq:P1}
\begin{array}{llll}
\mbox{\bf 1) Gaussian mean:}&\hspace*{.5in}\mbox{for }m\in\{0,1\}:\quad &\H_m:\quad & X_j^i\sim\mathcal{N}(\mu_m,1)\ .\\
\mbox{\bf 2) Gaussian variance:}&\hspace*{.5in}\mbox{for }m\in\{0,1\}:\quad &\H_m:\quad & X_j^i\sim\mathcal{N}(0,A_m)\ .
\end{array}
\end{eqnarray}
In this paper for the asymptotically large values of $n$ we characterize:
\begin{enumerate}
\item the optimal stopping time and sampling process (through designing the optimal switching sequence) that maximize the detection reliability subject to the {hard} constraints on the sampling resources and refinement actions;
\item the asymptotic detection error $\sP_n(S,K)$ for any given $S$ and $K$;
\item the minimum distance between the unknown pair of parameters in each test such that the two hypotheses are guaranteed to be distinguished perfectly, i.e., $\sP_n(S,K) \xrightarrow{\tiny n\rightarrow\infty}0$;
\item efficiencies of the proposed sequential and adaptive sampling procedure with respect to that of the non-adaptive procedure that does not involve any refinement action and performs detection directly, i.e., $K=0$; and
\item comparisons with a few relevant tests.
\end{enumerate}
Note that among the three actions, action A$_2$ (observation) concentrates on accumulating further evidence about the statistical behavior of the sequences under observation. Action A$_3$ (refinement) is intended to purify the set of the sequences on which we ultimately perform the detection action. Each iteration of the refinement actions  monitors the events retained by the previous iterations and eliminates those deemed least-likely to be rare events. The core idea of the refinement action is that it is relatively easy to identify sequences drawn from $F_1$ with low-quality measurements as the events of interest occur rarely ($\epsilon_n$ is small). After iteratively performing actions A$_2$ and A$_3$, we will have a more condensed proportion of the desired sequences to the non-desired ones. The sequences retained after the last action A$_3$ are further observed and finally fed to the detector in order to identify $T_n$ sequences of interest.

Characterizing the detection error rate $\PP(\tau,\bar\psi(\tau))$, as a result, depends on 1) the evolution of the proportion of the number of the normal and rare events throughout the refinement actions, and 2) the quality of the detector. For any time $t\in\{1,\dots,\tau\}$ let us partition the set $\L_t$, which includes the indices of the sequences observed at time $t$, into two disjoint sets $\L^0_t$ and $\L^1_t$ as
\begin{equation*}
\L^m_t\dff\{i\in\L_t:\;\T_i=\H_m\},\quad\mbox{for }\;\;m\in\{0,1\}\ .
\end{equation*}
Also let us define $n_t\dff|\L^1_t|$ and $\bar{n}_t\dff|\L^0_t|$ for all $t\in\{1,\dots,\tau\}$. As will be made clear in the subsequent sections, implementing action $\rm A_1$ (detection) necessitates obtaining the low-order statistics of the sets $\{\Lambda^{\tau}_i:\;i\in\L^0_t\}$ and $\{\Lambda^{\tau}_i:\;i\in\L^1_t\}$ with sample sizes $n_{\tau}$ and $\bar{n}_{\tau}$, respectively. Evaluating the dynamics of the refinement action ($\rm A_3$) additionally requires analyzing the high-order statistics of the sets $\{\Lambda_t^i:\;i\in\L^0_t\}$ and the low-order statistics of the sets $\{\Lambda_t^i:\;i\in\L^1_t\}$ for all $t\in\{1,\dots, \tau\}$. Moreover, $n_t$ and $\bar{n}_t$ are also random variables that change (reduce) after each refinement action. Analyzing the evolution of the set sizes $n_t$ and $\bar{n}_t$ as well as the necessary order statistics depends on the distributions $F_0$ and $F_1$.

\section{Preliminaries}
\label{sec:pre}

In this section we provide some definitions, notation, and basic results on asymptotic statistical behavior of order statistics that we frequently use throughout the rest of the paper. Given a sample of $m$ random variables $Y_1,\dots,Y_m$, we denote the corresponding sequence of order statistics by $Y_{1:m},\dots,Y_{m:m}$, where $Y_{r:m}$ is the $r^{th}$ order statistic. When the $Y_i$'s are independent and drawn from the same parent distribution with cdf $G$, the cdf of $Y_{r:m}$, denoted by $G_{r:m}$, is given by
\begin{equation*}
    G_{r:m}(y)=\sum_{k=r}^m\;{m\choose k}\;[G(y)]^k\;[1-G(y)]^{m-k}\ .
\end{equation*}
When $m$ tends to infinity the limit distributions $G_{1:m}$ and $G_{m:m}$ become degenerate as
\begin{equation*}
\lim_{m\rightarrow\infty}G_{1:m}(y)=\mathds{1}_{\{G(y)>0\}}\quad \mbox{and}\quad \lim_{m\rightarrow\infty}G_{m:m}(y)=\mathds{1}_{\{G(y)<1\}}\ .
\end{equation*}
A common approach to avoid degeneracy is to identify two {\em appropriate} sets of affine transformations of $Y_1,\dots,Y_m$
\begin{eqnarray}\label{eq:WL}
W_i & = & a_m\;+\;b_m\;Y_i\ ,\quad\mbox{for}\quad i\in\{1,\dots,m\}\ ,\\
\label{eq:WU}\mbox{and}\quad
\bar W_i & = & c_m\;+\;d_m\;Y_i\ ,\quad\mbox{for}\quad i\in\{1,\dots,m\}\ ,
\end{eqnarray}
that have non-generate limit distributions. Specifically, the cdfs of the order statistics $W_{1:m}$ and $\bar W_{1:m}$, denoted by $Q_{1:m}(\cdot;m)$ and $\bar Q_{1:m}(\cdot;m)$, respectively, satisfy
\begin{eqnarray}\label{eq:L}
&&\lim_{m\rightarrow\infty}Q_{1:m}(w;m)\;=\;L(w)\ ,\quad \forall w\in\mathbb{R}\ ,
\\
\label{eq:H2} \mbox{and}\quad
&&\lim_{m\rightarrow\infty}\bar Q_{1:m}(w;m)\;=\;H(w)\ ,\quad \forall w\in\mathbb{R}\ ,
\end{eqnarray}
for some non-degenerate cdfs $L$ and $H$. The precise characterizations of $L$ and $H$ depend on the transformation coefficients $a_m$, $b_m$, $c_m$, and $d_m$. The random variables $\{W_{i}\}_{i=1}^m$ and $\{\bar W_i\}_{i=1}^m$, in general, have completely different statistical behavior in low, central, and high order statistics. In the analysis in this paper we need the distributions of all kind of orders, which are formally defined next.
\begin{definition}\label{def:low_high}
The order statistic $Y_{r:m}$ is said to be of low order and the order statistic $Y_{m-r+1:m}$ is said to be of high order
iff $\lim_{m\rightarrow\infty}\frac{r}{m}=0$.
\end{definition}\label{def:central}
\begin{definition}
The order statistic $Y_{r:m}$ is said to be of central order iff $\exists\; \zeta\in(0,1)$ such that
\begin{equation*}
\lim_{m\rightarrow\infty}\sqrt{m}\left(\frac{r}{m}-\zeta\right)=0\ .
\end{equation*}
\end{definition}
\begin{definition}[Domain of attraction]
A given parent cdf $G$ for $\{Y_i\}_{i=1}^m$ is said to belong to the minimal domain of attraction of a cdf $L$ if there exists at least one pair of sequences $\{a_m\}_{m=1}^\infty$ and $\{b_m\}_{m=1}^\infty$ such that the sequence of random variables $\{W_i\}_{i=1}^m$ defined in \eqref{eq:WL} satisfies \eqref{eq:L}. Similarly, we say that $G$ belongs to the maximal domain of attraction of a cdf $H$ if for at least one pair of sequences $\{c_m\}_{m=1}^\infty$ and $\{d_m\}_{m=1}^\infty$ the sequence of random variables $\{\bar W_i\}_{i=1}^m$ defined in \eqref{eq:WU} satisfies \eqref{eq:H2}.
\end{definition}
Finding the possible cdfs $L$ and $H$ corresponding to any given parent distribution $F$ has a rich literature; cf. \cite{Fisher:Cambridge28} or \cite{Galambos:book}. The well-known result is that only one parametric family is possible for $L$ and only one family for $H$:
\begin{theorem}[von-Mises family of distribution]\label{thm:minima}
The only non-degenerate family of distributions satisfying \eqref{eq:L} is of the form
\begin{equation}\label{eq:Lk}
L(w)=1-\exp\left\{-\left[1+\kappa\left(\frac{w-\lambda}{\sigma}\right)\right]^{1/\kappa}\right\},\quad 1+\kappa\left(\frac{w-\lambda}{\sigma}\right)\geq 0\ ,
\end{equation}
and the only non-degenerate family of distributions satisfying \eqref{eq:H2} is of the form
\begin{equation}\label{eq:Hk}
H(w)=\exp\left\{-\left[1-\kappa\left(\frac{w-\lambda}{\sigma}\right)\right]^{1/\kappa}\right\},\quad 1-\kappa\left(\frac{w-\lambda}{\sigma}\right)\geq 0\ .
\end{equation}
\end{theorem}
\begin{proof}
See \cite{Fisher:Cambridge28, Galambos:book}.
\end{proof}
Besides the limit distributions of minima and maxima, in this paper we also need to find the asymptotic non-degenerate distributions corresponding to low and central order statistics. The following theorems characterize the asymptotic non-degenerate distributions for these order statistics as functions of the limit distributions of minima.
\begin{theorem}[Asymptotic distributions of low-order statistics]\label{th:low}
Let $G$ be a continuous cdf that belongs to the minimal domain of attraction of a cdf $L$ with the associated pair of sequences $\{a_m\}_{m=1}^\infty$ and $\{b_m\}_{m=1}^\infty$. Then for $r=o(m)$ we have
\begin{equation}\label{eq:Lr_dist}
\lim_{m\rightarrow\infty} Q_{r:m}(w;m)=
1-\mathds{1}_{\{L(w)<1\}}\cdot [1-L(w)]\sum_{i=0}^{r-1}\frac{[-\ln(1-L(w))]^i}{i!}\ ,\quad\forall w\in\mathbb{R}\ .
\end{equation}
\end{theorem}
\begin{proof}
See \cite[Theorem 8.4.1]{Arnold:Book}.
\end{proof}
\begin{theorem}[Asymptotic distributions of central order statistics]\label{th:central}
Let a continuous cdf $G$ with associated pdf $g$ be the distribution of $m$ i.i.d. random variables. Let $\zeta$ be a real number in the interval $(0,1)$ such that $g(G^{-1}(\zeta))\neq 0$. Then central order statistic $r=\lceil m\zeta\rceil$ is distributed as
\begin{equation*}
Y_{r:m}\sim \mathcal{N}\left(G^{-1}(\zeta)\;,\;\frac{\zeta(1-\zeta)}{m\left[g(G^{-1}(\zeta))\right]^2}\right)\ .
\end{equation*}
\end{theorem}
\begin{proof}
See \cite[Theorem 8.5.1]{Arnold:Book}.
\end{proof}
Finally, we define asymptotic equivalence and equality as follows.
\begin{definition}[Asymptotic equivalence]\label{def:asym}
Two sequences $\{a_m\}$ and $\{b_m\}$ are said to be {\em asymptotically equivalent}, denoted by $a_m\aeq b_m$, when $\lim_{m\rightarrow\infty}\frac{a_m}{b_m}=1$. We also say that $\{a_m\}$ and $\{b_m\}$ are {\em asymptotically equal}, denoted by $a_m\doteq b_m$, when $\lim_{n\rightarrow\infty}(a_n-b_n)=0$.
\end{definition}
\begin{definition}[Asymptotic dominance]
A Sequence $\{a_m\}$ is said to be asymptotically dominated by $\{b_m\}$, denoted by $a_m=o(b_m)$, when $\lim_{m\rightarrow\infty}\frac{a_m}{b_m}=0$. Also, $\{a_m\}$ is said to be asymptotically dominating $\{b_m\}$, denoted by $a_m=\omega(b_m)$, when $b_m=o(a_m)$
\end{definition}
\section{Optimal Sampling}
\label{sec:optimal}
In this section we determine the optimal stopping time $\tau$ and the associated optimal switching sequence $\bar\psi(\tau)=\{\psi(1),\dots,\dots,\psi(\tau-1)\}$. For this purpose, we initially characterize the error probability $\PP(\tau,\bar\psi(\tau))$ defined in \eqref{eq:Pn} for any given stopping time $\tau$ and switching sequence $\bar\psi(\tau)$ and then optimize it over the valid choices of $\tau$ and $\bar\psi(\tau)$ such that the constraints on the aggregate sampling budget and the number of refinement actions are satisfied. As stated earlier, characterizing the detection performance $\PP(\tau,\bar\psi(\tau))$ for a given stopping time $\tau$ and a switching sequence $\bar\psi(\tau)$ involves analyzing 1) the evolution of the proportion of the number of sequences generated by $F_0$ and $F_1$ throughout the refinement actions, and 2) the quality of the detection action. Analyzing both aspects heavily depend on the underlying distributions $F_0$ and $F_1$ and are treated separately for two the problems given in \eqref{eq:H_Gaussian}.
\subsection{Gaussian Mean}
\label{sec:mean}
The hypothesis-testing problem in this case is
\begin{equation}\label{eq:H_mean}
    \begin{array}{cc}
      \H_0: & X^i_j\sim \mathcal{N}(\mu_0,1), \quad j=1,2,\dots\\
      \H_1: & X^i_j\sim \mathcal{N}(\mu_1,1), \quad j=1,2,\dots
    \end{array}
\end{equation}
where $\mu_0>\mu_1$ and $\mu_0,\mu_1\in\mathbb{R}$. Under this setting, the likelihood ratio at time $t$ for the sequences $i\in\L_t$ is
\begin{equation}\label{eq:Lambda_2}
\Lambda_t^i \;=\; \prod_{u=1}^t\frac{f_0(X_u^i)}{f_1(X_u^i)}\;=\; \exp\left\{(\mu_0-\mu_1)\sum_{u=1}^t\;X_u^i\right\}\cdot \prod_{u=1}^t\exp\left\{\frac{\mu_1^2-\mu_0^2}{2}\right\} \ .
\end{equation}
By defining
\begin{equation}\label{eq:Z}
    \forall t\in\{1,\dots,\tau\}\;\;\mbox{and }\;\forall i\in\L_t:\quad Z_t^i\dff\sum_{u=1}^tX_u^i\ ,
\end{equation}
the detection and refinement actions formalized in \eqref{eq:U} and \eqref{eq:A3}, respectively, are equivalently given by
\begin{eqnarray}\label{eq:U2}
    \U&=&\mbox{the indices of the $T_n$ smallest elements of the set } \{Z^i_{\tau}:\;i\in\L_{\tau}\}\ ,\mbox{ and}\\
    \L_{t+1} & = &\mbox{the indices of the $\lfloor\alpha|\L_t|\rfloor$ smallest elements of the set } \{Z^i_{t}:\;i\in\L_{t}\}\ .
\end{eqnarray}
On the other hand, by invoking the distribution of $X_t^i$ from  \eqref{eq:H_mean} we immediately have
\begin{equation}\label{eq:Z_dist}
    \forall t\in\{1,\dots,\tau\}\;\;\mbox{and }\;\forall i\in\L_t:\quad Z_t^i\med \H_m\;\sim\;{\cal N}\Big(\mu_m\cdot t\;,\;t\Big)\ .
\end{equation}
Furthermore, let us for each $t\in\{1,\dots,\tau\}$ define
\begin{eqnarray}
\label{eq:U0}
&&\forall j\in\{1,\dots, n_t\}:\quad \bar U^t_j\;\dff\;\mbox{the $j^{th}$ smallest element of } \{Z_t^i:\;i\in\L^0_t\}\ ,\\
\label{eq:U1}
\mbox{and} && \forall j\in\{1,\dots, \bar{n}_t\}:\quad U^t_j\;\dff\;\mbox{the $j^{th}$ smallest element of } \{Z_t^i:\;i\in\L^1_t\}\ .
\end{eqnarray}
Given these definitions, the probability $\PP(\tau,\bar\psi(\tau))$, which is the probability that at least one member of $\U$ is distributed according to $F_1$, can be equivalently written as
\begin{equation}\label{eq:Pn2}
\PP_n(\tau,\bar\psi(\tau))=\P\;(\;|\;\U\cap \L^0_{\tau}\;|\;\geq 0)=\P\;\Big(\;U^{\tau}_{T_n}> \bar U^{\tau}_1\;\Big)\ .
\end{equation}
Assessing this performance measure involves finding the cardinalities and the distributions of the first and the $T_n^{th}$ order statistics of the sets of random variables $\{Z_{\tau}^i:\;i\in\L^0_t\}$ and $\{Z_{\tau}^i:\;i\in\L^1_t\}$, respectively. In order to proceed with analyzing $\PP_n(\tau,\bar\psi(\tau))$, we need the distributions of these order statistics. As the first step, the following lemma determines to what minimal domain of attraction the standard Gaussian distribution belongs. Also, this lemma in conjunction with Theorem~\ref{th:central} determines the distribution of the $T_n^{th}$ order statistic.
\begin{lemma}[Gaussian Minimum]\label{lmm:gaussian_min}
The standard Gaussian distribution belongs to the minimal domain of attraction of
\begin{equation*}
    L(w)=1-\exp\left(-\exp(w-\ln 2\sqrt{\pi})\right)\ ,\qquad\forall w\in\mathbb{R}\ ,
\end{equation*}
corresponding to \eqref{eq:Lk} for $\kappa\rightarrow 0$. The associated affine transformation is characterized by the sequences $\{a_m\}$ and $\{b_m\}$ given b y
\begin{equation*}
 a_m\;\dff\;h(m)\ ,\quad\mbox{and}\quad b_m\;\dff\;\sqrt{h(m)}\ ,\qquad\forall m\in\{2,3,\dots\}\ ,
\end{equation*}
where $h:\{x\in\mathbb{R}\;:\;x>1\}\rightarrow\mathbb{R}^+$ is defined as
\begin{equation}\label{eq:h}
        h(x)=2\ln x-\ln\ln x\ .
\end{equation}
\end{lemma}
\begin{proof}
See Appendix \ref{app:lmm:gaussian_min}.
\end{proof}
In the first step we assess the variations of $n_t$, $\bar n_t$, and their associated ratio throughout the refinement actions. The following lemma sheds light on the tendency of the refinement action (A$_2$) towards retaining {\em almost} all the rare events and discarding {\em almost} $\alpha$ proportion of the normal events.
\begin{lemma}[Refinement Performance]\label{lmm:refinement}
Let $\bar n_t=|\L^0_t|$ and ${n}_t=|\L^1_t|$ denote the number of sequences generated by $F_0$ and $F_1$ that are retained up to time $t$. For any arbitrary $\delta\in (0,1)$ and for sufficiently large $n$, the event
\begin{equation}\label{eq:n_t}
   n_{\tau}\;\geq\; (1-\delta)n_1\
\end{equation}
holds almost surely if
\begin{equation}\label{eq:mu_scaling_ratio0}
    \mu_0-\mu_1 \;=\; \omega\left(n^{-{\varepsilon_n}/2}\right)\ ,
\end{equation}
where $\varepsilon_n$ is defined as
\begin{equation}
{\varepsilon_n}\dff\frac{\ln n\epsilon_n}{\ln n}\ .
\end{equation}
\end{lemma}
\begin{proof}
See Appendix \ref{app:lmm:refinement}.
\end{proof}
Therefore, when the condition in \eqref{eq:mu_scaling_ratio0} is satisfied, the refinement actions almost surely discard no more than a fraction $\delta$ of the rare events, for any arbitrary $\delta\in(0,1)$. Therefore, the final ratio of the number of rare events to that of the normal events increases dramatically throughout the refinement actions. Besides the performance of the refinement actions, the overall detection reliability also depends on the performance of the detection action (A$_1$). The next lemma describes a necessary and sufficient condition that guarantees asymptotically error-free detection.
\begin{lemma}[Detection Performance]\label{lmm:detection}
For a given stopping time $\tau$ and switching sequence $\bar\psi(\tau)$, the detection error probability $\PP_n(\tau,\bar\psi(\tau))$ tends to zero in the asymptote of large $n$ if and only if
\begin{equation}\label{eq:scaling_mean0}
   r_{\rm m}\;>\;\frac{(1-\sqrt{{\varepsilon_n}})^2}{\tau}\ ,
\end{equation}
where we have defined
\begin{equation}\label{eq:rm}
    r_{\rm m}\dff \frac{(\mu_0-\mu_1)^2}{2\ln n}\ .
\end{equation}
\end{lemma}
\begin{proof}
See Appendix \ref{app:lmm:detection}.
\end{proof}
Therefore, this lemma, conditionally on the value of the stopping time which is stochastic, provides a necessary and sufficient condition on the distance between distributions $F_0$ and $F_1$ (captured by $(\mu_0-\mu_1)$) for achieving asymptotically optimal detection performance. In the next lemma we show that the stochastic stopping time is upper bounded by a constant.
\begin{lemma}\label{lmm:tau}
The stopping time $\tau$ is upper bounded by $S/\alpha^K$ for sufficiently large $n$.
\end{lemma}
\begin{proof}
See Appendix \ref{app:lmm:tau}.
\end{proof}
Combining the results of Lemmas \ref{lmm:detection} and \ref{lmm:tau} and replacing the stopping time $\tau$ in \eqref{eq:scaling_mean0} with its upper bound from Lemma~\ref{lmm:tau} provides that
\begin{equation*}
\frac{(\mu_0-\mu_1)^2}{2\ln n}\;>\;\frac{(1-\sqrt{{\varepsilon_n}})^2}{S/\alpha^K}\ ,
\end{equation*}
is a {\em necessary} condition for ensuring asymptotically error-free detection. This clearly imposes a more stringent condition on the distance $(\mu_0-\mu_1)$ than \eqref{eq:mu_scaling_ratio0}, which is a {\em sufficient} condition for retaining at least a fraction $(1-\delta)$ of the rare events for throughout the refinement cycles. As a result, irrespective of the optimal value of the stopping time $\tau$, ensuring asymptotically error-free detection requires the refinement process to retain at least $(1-\delta)$ fraction of the rare events.

Finally, we would like to remark that while the exact characterization of the error probability $\PP_n(\tau,\bar\psi(\tau))$, {\em does} depend on $T_n$ (the number of events to be returned by the detector), the scaling rate given in \eqref{eq:scaling_mean0} {\em does not} depend on $T_n$. The intuitive reason is that in analyzing such scaling laws, in the asymptote of large $n$, only the dominant (leading) terms survive and the impacts of the non-leading terms vanish as $n$ grows. $T_n$ appears to impact only the non-leading terms and its impact is not observed on the ultimate scaling laws. In a more technical sense, the reason  lies at the core of the analysis on order statistics. Note that the focus of this paper is on the regime $T_n=o(n\epsilon_n)$. The analysis reveals that while the error probabilities vary for different choices of $T_n$, the scaling laws of the mean values is identical for all values of $T_n$ in the regime $T_n=o(n\epsilon_n)$. The underlying reason for this observation is that according to the definition of the detection action, the ultimate $T_n$ sequences selected are the sequences with the smallest likelihood ratios, and hence, constitute the $T_n$ smallest order statistics of the set comprised of the likelihood ratios of the retained sequences. When $T_n=o(n\epsilon_n)$ all these order statistics are of low order (Definition \ref{def:low_high}) and have identical asymptotic impacts on the scaling rates.


\subsection{Gaussian Variance}
\label{sec:var}
In this section we analyze the performance of detection and refinement in the Gaussian variance hypothesis testing problem. The presentation of the results follows the same flow as the problem of testing the mean. The proofs, however, are different due to different statistical behavior of the likelihood ratio and its pertinent sufficient statistic. The problem of interest can be posed as
\begin{equation}\label{eq:H_var}
    \begin{array}{cc}
      \H_0: & X^i_j\sim \mathcal{N}(0, A_0), \quad j=1,2,\dots\\
      \H_1: & X^i_j\sim \mathcal{N}(0, A_1), \quad j=1,2,\dots
    \end{array}
\end{equation}
where $A_0>A_1$ and $A_0,A_1\in\mathbb{R}^+$. Hence, the likelihood ratio at time $t$ for the sequences $i\in\L_t$ is
\begin{equation}\label{eq:Lambda_2_var}
\Lambda_t^i \;=\; \prod_{u=1}^t\frac{f_0(X_u^i)}{f_1(X_u^i)}\;=\; \exp\left\{\frac{1}{2}\left(\frac{1}{A_1}-\frac{1}{A_0}\right)\sum_{u=1}^t\;(X_u^i)^2\right\}\ .
\end{equation}
Similarly to \eqref{eq:Z} we define
\begin{equation}\label{eq:Z_var}
    \forall t\in\{1,\dots,\tau\}\;\;\mbox{and }\;\forall i\in\L_t:\quad \bar Z_t^i\dff\sum_{u=1}^t(X_u^i)^2\ .
\end{equation}
Therefore, the indices of the sequences yielded by the detection and refinement actions defined in \eqref{eq:U} and \eqref{eq:A3}, respectively, are given by
\begin{eqnarray}\label{eq:U2_var}
    \U&=&\mbox{the indices of the $T_n$ smallest elements of the set } \{\bar Z^i_{\tau}:\;i\in\L_{\tau}\}\ ,\mbox{ and}\\
    \L_{t+1} & = &\mbox{the indices of the $\bar\alpha$ smallest elements of the set } \{\bar Z^i_{t}:\;i\in\L_{t}\}\ ,
\end{eqnarray}
where $\bar\alpha$ is defined in \eqref{eq:alpha_bar}.
By recalling the distribution of $X_t^i$ given in~\eqref{eq:H_var} we have
\begin{equation}\label{eq:Z_dist_var}
    \forall t\in\{1,\dots,\tau\}\;\;\mbox{and }\;\forall i\in\L_t:\quad \bar Z_t^i\med \H_m\;\sim\;A_m\cdot\chi^2(t)\ ,
\end{equation}
where $\chi^2(k)$ denotes a chi-squared distribution with $k$ degrees of freedom. Next, for each $t\in\{1,\dots,\tau\}$ let us define
\begin{eqnarray}
\label{eq:U0_var}
&&\forall j\in\{1,\dots, \bar{n}_t\}:\quad \bar V^t_j\dff\mbox{the $j^{th}$ smallest element of } \{\bar Z_t^i:\;i\in\L^0_t\}\ ,\\
\label{eq:U1_var}
\mbox{and} && \forall j\in\{1,\dots, n_t\}:\quad V^t_j\dff\mbox{the $j^{th}$ smallest element of } \{\bar Z_t^i:\;i\in\L^1_t\}\ .
\end{eqnarray}
As shown in \eqref{eq:Pn2} the probability $\PP(\tau,\bar\psi(\tau))$ is given by
\begin{equation}\label{eq:Pn2_var}
\PP_n(\tau,\bar\psi(\tau))=\P\;(\;|\;\U\cap \L^1_{\tau}\;|\;\geq\; 0\;)=\P\;\Big(\;V^{\tau}_{T_n}\;>\; \bar V^{\tau}_1\;\Big)\ .
\end{equation}
As assessing this probability requires knowing the distributions of the first and the $T_n^{th}$ order statistics of the sets of random variables $\{\bar Z_{\tau}^i:\;i\in\L^0_t\}$ and $\{\bar Z_{\tau}^i:\;i\in\L^1_t\}$, respectively, in the following lemma we provide the minimal domain of attraction of chi-squared distributions.
\begin{lemma}[$\chi^2$ Minimum]\label{lmm:chi_min}
The chi-squared distribution with $k$ degrees of freedom belongs to the minimal domain of attraction of
\begin{equation*}
    L(w)=1-\exp\left(-w^{k/2}\right)\ ,\qquad\forall w\in\mathbb{R}\ ,
\end{equation*}
which is corresponding to \eqref{eq:Lk} for $\kappa\rightarrow 0$.
The associated affine transformation is characterized by the sequences $\{a_m\}$ and $\{b_m\}$ given by
\begin{equation*}
 a_m=0\ ,\quad\mbox{and}\quad b_m=\frac{1}{2}\cdot\left[\frac{m}{\Gamma(k/2+1)}\right]^{2/k},\qquad\forall m\in\{2,3,\dots\}\ .
\end{equation*}
\end{lemma}
\begin{proof}
See Appendix \ref{app:lmm:chi_min}.
\end{proof}
Also, as will be shown later, for analyzing the performance of the detection action we also need to determine to what maximal domain of attraction a $\chi^2$ distribution belongs.
\begin{lemma}[$\chi^2$ Maximum]\label{lmm:chi_max}
The chi-squared distribution with $k$ degrees of freedom belongs to the maximal domain of attraction of
\begin{equation*}
    H(w)=\exp\left(-\frac{1}{\Gamma(k/2)}\cdot \exp(-w)\right)\ ,\qquad\forall w\in\mathbb{R}\ ,
\end{equation*}
which is corresponding to \eqref{eq:Hk} for $\kappa\rightarrow 0$.
The associated affine transformation is characterized by the sequences $\{a_m\}$ and $\{b_m\}$ given by
\begin{equation*}
 c_m= -\left(\ln m+\frac{K-1}{2}\ln\ln m\right)\ , \quad\mbox{and}\quad d_m=\frac{1}{2} ,\qquad\forall m\in\{2,3,\dots\}\ .
\end{equation*}
\end{lemma}
\begin{proof}
See Appendix \ref{app:lmm:chi_max}.
\end{proof}
In the next lemma we provide the performance of the refinement action in retaining the sequences generated by $F_0$ and $F_1$.
\begin{lemma}[Refinement Performance]\label{lmm:refinement_var}
Let $\bar n_t=|\L^0_t|$ and ${n}_t=|\L^1_t|$ denote the number of sequences generated by $F_0$ and $F_1$, respectively, that are retained up to time $t$. For sufficiently large $n$, the event
\begin{equation}\label{eq:n_t_var}
   n_\tau\;=\;  n_{1}\
\end{equation}
holds almost surely if
\begin{eqnarray}\label{eq:A_scaling_ratio}
\frac{A_0}{A_1}\;=\;\omega\;({\varepsilon_n}\ln n)\ ,
\end{eqnarray}
where $\varepsilon_n$ is defined as
\begin{equation}
{\varepsilon_n}\dff\frac{\ln n\epsilon_n}{\ln n}\ .
\end{equation}
\end{lemma}
\begin{proof}
See Appendix \ref{app:lmm:refinement_var}.
\end{proof}
Therefore, when the scaling law in \eqref{eq:A_scaling_ratio} is satisfied, the proportion of the sequences generated by $F_0$ to those generated by $F_1$ increases after the refinement actions. In the next lemma we provide a necessary and sufficient condition on the scaling of $A_1/A_0$ that ensures perfect identification of the $T_n$ sequences generated by $F_1$ (rare events).
\begin{lemma}[Detection Performance]\label{lmm:detection_var}
For a given stopping time $\tau$ and switching sequence $\bar\psi(\tau)$, and conditionally on $n_{\tau}$ and $\bar{n}_{\tau}$, the detection error probability $\PP_n(\tau,\bar\psi(\tau))$ tends to zero in the asymptote of large $n$ if and only if
\begin{equation}\label{eq:scaling_var0}
   \xi_{\rm v}\; >\; \frac{2(1-{\varepsilon_n})}{\tau}\ ,
\end{equation}
where we have defined
\begin{equation}\label{eq:xi}
    \xi_{\rm v}\dff\frac{\ln\frac{A_0}{A_1}}{\ln n}\ .
\end{equation}
\end{lemma}
\begin{proof}
See Appendix \ref{app:lmm:detection_var}.
\end{proof}
By comparing the scaling laws offered by Lemmas~\ref{lmm:refinement_var}~and~\ref{lmm:detection_var} we find that the scaling law {\em necessary} for making a reliable detection, irrespective of $\tau$ and $\bar\psi(\tau)$, dominates the one that is {\em sufficient} for maintaining $n_\tau=n_1$ almost surely. In other words, in order to perform reliable detection the refinement action (A$_2$) retains all the rare events almost surely.
\subsection{Optimal Stopping Time}
\label{sec:stop}
Given the performance of the refinement action offered by Lemmas \ref{lmm:refinement}~and~\ref{lmm:refinement_var}, and the detection action given by Lemmas~\ref{lmm:detection}~and~\ref{lmm:detection_var}, in this section we provide the optimal choices of the stopping time and the switching sequence. Given the discussions at the end of Sections~\ref{sec:mean}~and~\ref{sec:var}, irrespective of the discrepancies in the analysis and the ensuring scaling laws in the mean and variance settings,  these two settings conform in the fact that targeting at error-free detection forces the refinement process to retain {\em almost} all of the rare events. More specifically, the refinement process is guaranteed (probabilistically) to discard at most a fraction $\delta$ of the rare events in testing the mean and retain almost all such events in testing the variance.

Primarily based on the similar behavior of the refinement process in both settings, the optimal choices of the stopping time and the switching sequence turn out to be exactly the same in both settings. The scheme of the proofs in this section are as follows. By using the result of Lemma~\ref{lmm:tau}, which states that the stopping time $\tau$ is upper bounded by a constant, it is concluded that regardless of the choice of the stopping time and the switching sequence, the conditions that Lemmas~\ref{lmm:detection}~and~\ref{lmm:detection_var} impose on $(\mu_0-\mu_1)$ and $A_1/A_0$, respectively, dominate those imposed by Lemmas~\ref{lmm:refinement}~and~\ref{lmm:refinement_var}, respectively. This immediately indicates that the refinement process in both settings retains almost all rare events almost surely. Based on this property of the refinement process, we obtain the choices of the stopping time and switching sequence that minimize the error probability.
\begin{theorem}[Stopping Time]\label{th:stop}
For achieving $\sP_n(S,K) \xrightarrow{\tiny n\rightarrow\infty}0$ the optimal switching sequence satisfies
    \begin{equation}\label{eq:switch}
    \forall t\in\{1,\dots,K^*\}:\quad \psi(1)=1,\qquad\mbox{and}\qquad \forall t>K^*:\quad \psi(t)=0\ ,
  \end{equation}
  where
  \begin{equation}\label{eq:K_optimal}
K^*=\left\{
\begin{array}{ll}
K\ , & \mbox{if}\;\;\alpha\leq 1-\frac{1}{S}\\
\;0\ , & \mbox{if}\;\;\alpha> 1-\frac{1}{S}\\
\end{array}\right.\ .
\end{equation}
Also the optimal stopping time is
\begin{equation}\label{eq:t_optimal}
\tau=\left\{
\begin{array}{ll}
\;K+ s(K) \ , & \mbox{if}\;\;\alpha\leq 1-\frac{1}{S}\\
\;S\ , & \mbox{if}\;\;\alpha> 1-\frac{1}{S}\\
\end{array}\right.\ ,
\end{equation}
where
\begin{equation}\label{eq:S_bar}
    s(K)\;\dff\; \left\lfloor S\cdot\alpha^{-K}+\frac{1-\alpha^{-K}}{1-\alpha}\right\rfloor\ .
\end{equation}

\end{theorem}
\begin{proof}
See Appendix \ref{app:th:mean}.
\end{proof}
This theorem demonstrates that when $\alpha$ is large (close to 1) the optimal sampling does not involve any refinement action and adaptive sampling does not offer any gain over the non-adaptive sampling procedure. However, for sufficiently large $S$, only for limited choices of $\alpha$ does adaptation in sampling offer no gain and for a wide range of $\alpha$ the optimal sampling procedure involves refinement actions and becomes adaptive. Throughout the remainder of this paper we focus on the regime $\alpha\leq 1-\frac{1}{S}$.

\section{Adaptation Gains}
\label{sec:gain}

The main component of the proposed sampling procedure is the inclusion of the refinement actions, which makes the detection procedure adaptive to the data. In this section we show the gains attained by the inclusion of the refinement process. These gains can be viewed as the gain of adapting the detection process to the observed data and can be interpreted in two ways, namely in the forms of agility and scaling gains defined in the sequel. These gains essentially evaluate the contribution of the refinement actions by comparing the performance of the proposed sampling procedure against the same procedure without any refinement action, i.e., $K=0$.
\begin{definition}[Agility Gain] The agility gain, denoted by $G_{\rm a}$, is the ratio of the minimum sampling budget required by the non-adaptive procedure i.e., $K=0$, to that required by the adaptive procedure with $K>0$ refinement actions such that both achieve asymptotically error-free detection while enjoying identical scaling for $(\mu_0-\mu_1)$ in the mean setting or identical scaling for $A_0/A_1$ in the variance setting.
\end{definition}
In order to quantify this gain in the Gaussian mean setting we consider a non-adaptive detection procedure with the aggregate sampling budget controlled by $S_0$ and obtain the required scaling law for $(\mu_0-\mu_1)$ that guarantees asymptotically error-free detection. For the same scaling law in an adaptive procedure with $K$ refinement actions we assess the minimum sampling budget $S$ that ensures asymptotically error-free detection by the adaptive procedure as well. Then we find the agility gain as the ratio $\frac{S_0}{S}$. In order to analyze the agility gain for the variance setting we repeat the same steps by replacing all the arguments on the scaling of $(\mu_0-\mu_1)$ by those on the scaling of $A_0/A_1$.
\begin{definition}[Scaling Gain] The scaling gain, denoted by $G_{\rm s}$, is the ratio of the smallest scaling law required by the non-adaptive procedure to that required by the adaptive procedure when both target at achieving asymptotically error-free detection while enjoying identical sampling budgets.
\end{definition}
In order to analyze the scaling gain of the adaptive procedure with $K$ refinement actions in the Gaussian mean setting we assume that both adaptive and non-adaptive procedures are allocated the aggregate sampling budget $S\cdot n$ and aim at identifying the smallest values of $r_{\rm m}$ and $r_{\rm m}^0$ such that the scaling laws $(\mu_0-\mu_1)=\sqrt{2\;r_{\rm m}\ln n}\;$ and $(\mu_0-\mu_1)=\sqrt{2\;r^0_{\rm m}\ln n}$ guarantee asymptotically error-free detection for the adaptive and the non-adaptive procedures, respectively. The scaling gain can be found as $\frac{r^0_{\rm m}}{r_{\rm m}}$. For computing the scaling gain in the variance setting we follow the same logic and instead of finding $r_{\rm m}$ and $r_{\rm m}^0$ we consider the scaling laws $A_0/A_1=n^{\xi_{\rm v}}$ and $A_0/A_1=n^{\xi^0_{\rm v}}$ for the adaptive and non-adaptive procedures and find $\xi_{\rm v}$ and $\xi_{\rm v}^0$, respectively. In this case the scaling gain is defined as $\frac{\xi_{\rm v}^0}{\xi_{\rm v}}$.

Analyzing the scaling and agility gains strongly relies on the connection between the available sampling budget $S$ and the scaling laws that lead to reliable detection of the sequences of interest. The following two theorems establish this connection.
\begin{theorem}[Mean Scaling]
\label{th:power_NA} When $\alpha\leq 1-\frac{1}{S}$ a necessary and sufficient condition for $\sP_n(S,K)\xrightarrow{n\rightarrow\infty} 0$ is that
\begin{equation}\label{eq:mean_scaling}
r_{\rm m}\;>\;
\frac{(1-\sqrt{{\varepsilon_n}})^2}{K+s(K)}\ ,
\end{equation}
where $r_{\rm m}$ is defined in \eqref{eq:rm} and $s(K)$ was defined as
\begin{equation}\label{eq:S_bar_2}
    s(K)\;\dff\; \left\lfloor S\cdot\alpha^{-K}+\frac{1-\alpha^{-K}}{1-\alpha}\right\rfloor\ ,
\end{equation}
where $\alpha$ controls what fraction of the events are discarded at each refinement cycle, $S$ is the aggregate sampling budget normalized by $n$, and $K$ is the maximum allowable number of refinement actions.
\end{theorem}
\begin{proof}
The proof can be established by combining the results from Lemma~\ref{lmm:detection} and the optimal stopping time given in Theorem~\ref{th:stop}. As shown in the proof of Lemma~\ref{lmm:detection} a necessary and sufficient condition for having asymptotically error-free detection is that $B_n=\omega(1)$, where $B_n$ is defined in \eqref{eq:A}. Moreover, we have also proved that a necessary and sufficient condition for $B_n=\omega(1)$ is that
\begin{equation*}
r_{\rm m}\;>\;
\frac{(1-\sqrt{{\varepsilon_n}})^2}{\tau}\ .
\end{equation*}
On the other hand, in the proof of Theorem~\ref{th:stop} we showed that the detection error probability is minimized when $\tau$ is maximized. In other words, the smallest necessary scaling law for $(\mu_0-\mu_1)^2$ is obtained when $\tau$ is maximized and is equal to $K+s(K)$. By substituting this value into the equation above we obtain the desired result.
\end{proof}
For a given sampling budget $S\cdot n$ and the number of refinement actions $K$ this corollary delineates the asymptotic performance of the proposed detector in the $(r_{\rm m},{\varepsilon_n})$ plane.  It shows that when $(\mu_0-\mu_1)$ scales as $\sqrt{2\;r\ln n}$, if $r> r_{\rm m}$ the proposed sequential detection procedure is guaranteed to make error-free decisions for identifying $T_n$ sequences generated by $F_1$. On the other hand, when $r\leq r_{\rm m}$ the probability of erroneous detection is bounded away from zero. Therefore, $r_{\rm m}$ defines a sharp threshold for identifying the sequences of interest under the objective and constraints in \eqref{eq:P}.

On a related context, note the relevant results provided by ~\cite{Abramovich} and \cite{Jin:2006_1}  when the objective is to identify {\em all} sequences generated by $F_1$ through observing each sequence only {\em once}, i.e., $S=1$ and $K=0$. The works in~\cite{Abramovich} and \cite{Jin:2006_1} show that for ${\varepsilon_n}\in(0,\frac{1}{2})$ the false-discovery and non-discovery proportions\footnote{The false-discovery proportion is the number of falsely discovered sequences that are generated by $F_1$ relative to the total number of sequences, and the non-discovery proportion is the ratio of the number of sequences generated by $F_0$ that are missed to the entire number of sequence.} tend to zero if and only if $(\mu_0-\mu_1)$ scales as $\sqrt{2\;r\ln n}$ and $r>1-{\varepsilon_n}$. This is clearly a more stringent condition than the requirement $r>r_{\rm m}=(1-\sqrt{{\varepsilon_n}})$  corresponding to the setting $S=1$ and $K=0$ in our sampling procedure. This discrepancy demonstrates the tradeoff between {\em partial} versus {\em full} recovery of the sequences of interest on one hand, and the required scaling law on $(\mu_0-\mu_1)$ on the other hand.

In the following lemma we focus on the variance setting and by using the result of Lemma~\ref{lmm:detection_var} and Theorem~\ref{th:stop}, we offer a necessary and sufficient condition on the scaling of $A_0/A_1$ in order to guarantee asymptotically error-free detection.
\begin{theorem}[Variance Scaling]
\label{th:power_NA_var} When $\epsilon_n=o(1)$, $n\epsilon_n=\omega(1)$, and $\alpha\leq 1-\frac{1}{S}$ a necessary and sufficient condition for $\sP_n(S,K)\xrightarrow{n\rightarrow\infty} 0$ is that
\begin{equation}\label{eq:rv}
\xi_{\rm v}\;>\;
\frac{2(1-{\varepsilon_n})}{K+s(K)}\ ,
\end{equation}
where $\xi_{\rm v}$ is defined in \eqref{eq:xi} and $s(K)$ was defined as
\begin{equation}\label{eq:S_bar_2}
    s(K)\;\dff\; \left\lfloor S\cdot\alpha^{-K}+\frac{1-\alpha^{-K}}{1-\alpha}\right\rfloor\ ,
\end{equation}
where $\alpha$ controls what fraction of the events are discarded at each refinement cycle, $S$ is the aggregate sampling budget normalized by $n$, and $K$ is the maximum allowable number of refinement actions.
\end{theorem}
\begin{proof}
This results follows from the same line of argument as in the proof of Theorem~\ref{th:power_NA}.
\end{proof}
\begin{figure}[1t]\label{fig:1}
\centering
\subfigure[Mean scaling]{
\includegraphics[width = 3.1 in]{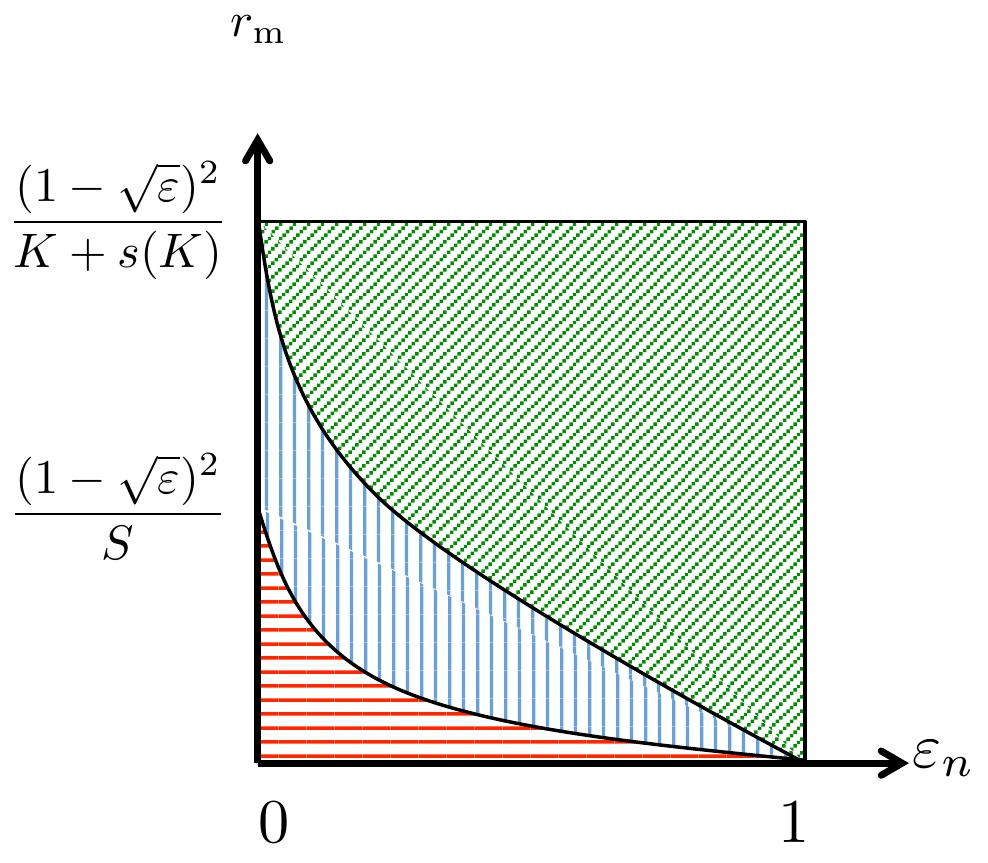}
\label{fig:subfig1}
}
\subfigure[Variance scaling]{
\includegraphics[width = 3.1 in]{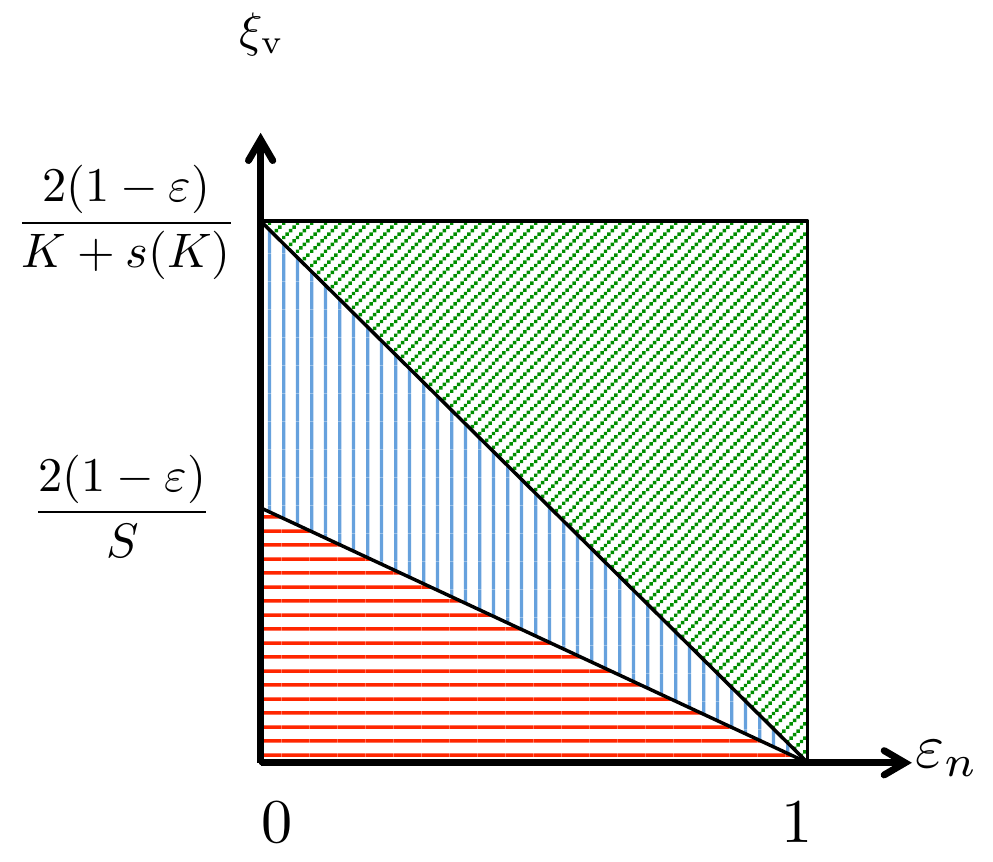}
\label{fig:subfig2}
}
\label{fig:subfigureExample}
\caption[Detectable regions]{Detectable vs. non-detectable regions}
\end{figure}

Figures \ref{fig:subfig1} compares the regions in the $(r_{\rm m},{\varepsilon_n})$ plane over which the adaptive and the non-adaptive procedures are guaranteed to make error-free detections. Specifically, the diagonally shaded region is the region in which both schemes succeed to detect the $T_n$ sequences of interest. In the vertically dashed region, however, only the adaptive procedure succeeds and the non-adaptive procedure makes an erroneous decision almost surely, and finally both schemes fail in the horizontally shaded region. It is observed that, depending on the choice of $K$, the detectability region corresponding to the adaptive procedure can be substantially larger than the one corresponding to the non-adaptive procedure. Figure~\ref{fig:subfig2} depicts the counterpart regions in the variance setting in the $(\xi_{\rm v},{\varepsilon_n})$ plane. Given  the necessary and sufficient conditions on the scaling law for performing error-free detection, in the following corollaries we determine the agility and scaling gains.
\begin{corollary}[Agility Gain]
\label{cor:agility_mean} When $\epsilon_n=o(1)$, $n\epsilon_n=\omega(1)$, and $\alpha\leq 1-\frac{1}{S}$ the agility gain satisfies
\begin{equation}\label{eq:Ga}
    \frac{1-\alpha^{K}}{1-\alpha}\;\leq\; S_0(G_{\rm a}^{-1}-\alpha^K) \;\leq\;\frac{1-\alpha^{K+1}}{1-\alpha}\ ,
\end{equation}
and in the asymptote of large $K$, the upper and lower bounds on the agility gain meet and are equal to $S_0(1-\alpha)$.
\end{corollary}
\begin{proof}
See Appendix \ref{app:cor:mean:agility}.
\end{proof}
This result indicates that for a given set of sequences, the proposed sequential detection procedure can make a detection decision with substantially fewer measurements compared with the non-adaptive sampling procedure. Specifically, the sampling budget required by the adaptive procedure reduces exponentially with the number of refinement actions $K$. It is noteworthy that while the number of refinement actions $K$ can be made arbitrarily large (but fixed as a function of $n$), increasing it beyond some point increases the agility gain marginally. More specifically, for large $K$, both upper and lower bounds on the agility gain tend to $S_0(1-\alpha)$, which is a constant. This observation sheds light on the fundamental limit of the agility gain yielded by adaptive sampling.

In Fig.~\ref{fig:QS} we provides a numerical comparison between the proposed adaptive sampling scheme and the CUSUM test repeated $T_n$ times. With the target error probability 
$\sP_n(S,K)=10^{-5}$, the plot depicts the necessary {\em average} number of samples by the CUSUM and the necessary the number of samples by the adaptive sampling schemes for different number of refinement actions. We consider $n=10^4$ events and set the prior probability $\epsilon_n=n^{\varepsilon-1}$ and aim to identify $T_n=\sqrt{n\epsilon_n}$ rare events. The normal events are distributed as ${\cal N}(0,A_0$ and the rare events as ${\cal N}(0,A_1)$ with the variance values satisfying $A0/A1=n^{1/20}$. It is observed that for the setting that distribution $F_1$ occurs very rarely (i.e., smaller values of $\varepsilon_n$) the adaptive sampling strategy is quicker than the repeated CUSUM tests. On the other hand, by increasing the priors, the repeated CUSUM test outperforms the adaptive sampling strategy.  

\begin{figure}[t]
\centering
\includegraphics[width=5in]{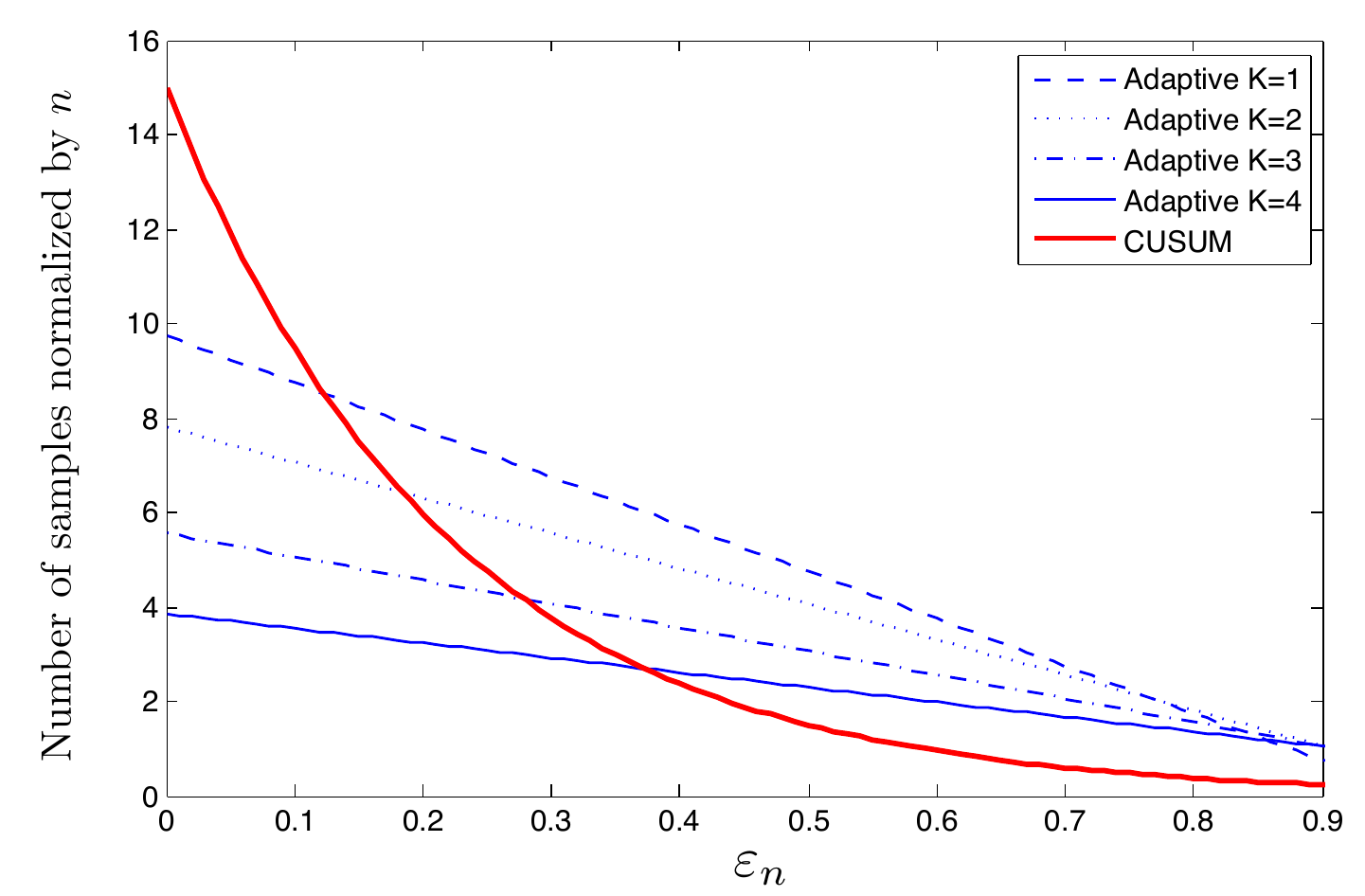}\\
\caption{Normalized sampling budget $S$ versus the prior probability controlled by $\varepsilon_n$.}
\label{fig:QS}
\end{figure}
\begin{corollary}[Scaling Gain]
\label{cor:scaling_mean} When $\alpha\leq 1-\frac{1}{S}$ the scaling gain is
\begin{equation*}
    \alpha^{-K}\left(1+\frac{1}{S}\cdot\frac{\alpha^{K+1}-1}{1-\alpha}\right)\;\leq\;{G_{\rm s}}\;\leq\;\alpha^{-K}\left(1+\frac{1}{S}\cdot\frac{\alpha^{K}-1}{1-\alpha}\right)\ ,
\end{equation*}
and in the asymptote of large $K$ the upper and lower bounds on $G_{\rm s}$ meet and are equal to $\alpha^{-K}(1-\frac{1}{S(1-\alpha)})$.
\end{corollary}
\begin{proof}
The proof follows the same line of argument as the proof of Corollary~\ref{cor:agility_mean} and the characterization of $\xi_{\rm v}$ in \eqref{eq:xi}.
\end{proof}
This corollary indicates that the scaling gain grows exponentially with the number of refinement actions. This result also indicates that the adaptive procedure can detect signals with much smaller means than those detectable by the non-adaptive procedure. More specifically, by noting that $\alpha^{-K}$ is substantially larger than $1$, the mean scaling requirement in the adaptive scenario becomes substantially less stringent than its counterpart in the non-adaptive scenario. As a result, there are scenarios in which non-adaptive schemes fail to successfully identify $T$ sequences of interest, while the adaptive scheme succeeds.

\section{Conclusion}
In this paper we have presented an adaptive sampling methodology for quickly searching over finitely many events with the objective of identifying multiple events that occur sparsely and are distributed according to a given distribution of interest. The main idea of the sampling procedure is to successively and gradually adjust the measurement process using information gleaned from the previous measurements. Compared to corresponding non-adaptive procedures, dramatic gains in terms of reliability and agility are achieved.

\appendix

\section{Proof of Lemma \ref{lmm:gaussian_min}}
\label{app:lmm:gaussian_min}
Let us define $Y_1,\dots,Y_m$ as i.i.d. standard Gaussian random variables with cdf $G$ and pdf $g$ and define the function $h:\{x\in\mathbb{R}\;:\;x>1\}\rightarrow\mathbb{R}^+$
as
\begin{equation*}
    h(x)\dff 2\ln x-\ln\ln x\ .
\end{equation*}
By defining
\begin{eqnarray}\label{eq:Wi_last}
    W_i & = & a_m\;+\;b_m\;Y_i\ ,
\end{eqnarray}
and setting
\begin{equation}\label{eq:ambm}
    a_m\;=\;h(m)\ ,\qquad\mbox{and}\qquad b_m=\sqrt{h(m)}\ ,
\end{equation}
for the cdf of $W_{i:m}$, denoted by $Q_{1:m}(\cdot;m)$, we have
\begin{eqnarray}\label{eq:1_Q}
  \nonumber 1-Q_{1:m}(w;m) &{=}& 1 - \P\left(W_{1:m}\leq w\right)\\
  \nonumber  & \overset{\eqref{eq:Wi_last}}{=} & 1- \P\left(Y_{1:m}\leq \frac{w-a_m}{b_m}\right)\\
  \nonumber &=& \left[1-G\left(\frac{w-a_m}{b_m}\right)\right]^m\\
  &\overset{\eqref{eq:ambm}}{=}& \left[1-G\left(\frac{w}{\sqrt{h(m)}}-\sqrt{h(m)}\right)\right]^m\ .
\end{eqnarray}
We next show that
\begin{equation}\label{eq:loglog}
    \lim_{x\rightarrow\infty}x\cdot \log\left[1-G\left(\frac{y}{\sqrt{h(x)}}-\sqrt{h(x)}\right)\right]=-\frac{e^w}{2\sqrt{\pi}}\ ,\qquad\forall w\in\mathbb{R}\ .
\end{equation}
By using L'H\^{o}pital's rule we obtain
\begin{align*}
 \lim_{x\rightarrow\infty}&\frac{\log\left[1-G\left(\frac{y}{\sqrt{h(x)}}-\sqrt{h(x)}\right)\right]}{\frac{1}{x}} \\
&=  \lim_{x\rightarrow\infty}\frac{\frac{1}{2}\;h'(x)\;\big[w\cdot h(x)\big]^{-1/2}\big[[h(x)]^{-1}+1\big]\cdot\frac{1}{1-G\left(\frac{w}{\sqrt{h(x)}}-\sqrt{h(x)}\right)}\cdot g\left(\frac{w}{\sqrt{h(x)}}-\sqrt{h(x)}\right)} {-\frac{1}{x^2}}\\
&\\
& = \lim_{x\rightarrow\infty}\frac{\frac{1}{2x\sqrt{2\pi}}\;\left[2-\frac{1}{\ln x}\right]\left[1+\frac{w}{2\ln x-\ln\ln x}\right]\exp\left(-\frac{w^2/2}{2\ln x-\ln\ln x}\right)\exp(w)\frac{\sqrt{\ln x}}{x}}{-\frac{1}{x^2}\;\sqrt{2\ln x-\ln\ln x}}\\
&\\
& = - \frac{e^w}{2\sqrt{\pi}}\ .
\end{align*}
Equations \eqref{eq:1_Q} and \eqref{eq:loglog} in conjunction with the continuity of $\ln(\cdot)$ establish that
\begin{eqnarray}
  \lim_{m\rightarrow\infty} Q_{1:m}(w;m)= 1- \exp\left(-\frac{e^w}{2\sqrt{\pi}}\right)\ ,
\end{eqnarray}
which is the desired result.

\section{Proof of Lemma \ref{lmm:refinement}}
\label{app:lmm:refinement}

From the definitions of $n_t$ and $\bar n_t$ for all $t\in\{1,\dots,\tau \}$ we have $|\L_t|=n_t+\bar n_t$. Therefore, for any time $t\in\{1,\dots,\tau-1 \}$ that the refinement action is taken (i.e., $\psi(t)=1$) the condition
\begin{equation*}
    \lceil(1-\bar\delta)n_{t}\rceil\;\leq\; n_{t+1}
\end{equation*}
for some $\bar\delta\in(0,1)$ is equivalent to having
\begin{equation*}
    \bar n_{t+1} \;\leq \; |\L_{t+1}|-\lceil(1-\tilde\delta)n_{t}\rceil\ .
\end{equation*}
By taking into account the distributions of the order statistics $\bar U^t_j$ and $U^t_j$ given in \eqref{eq:U0} and \eqref{eq:U1}, respectively, the event that the $|\L_{t+1}|$ retained events after a refinement action at time $t$ contains $\lceil(1-\tilde\delta)n_{t}\rceil$ rare events is equivalent to
\begin{equation}\label{eq:retain_event}
     U^t_{\lceil (1-\bar\delta)n_t\rceil}\;<\; \bar U^t_{|\L_{t+1}|-\lceil (1-\bar\delta)n_t\rceil+1}\ .
\end{equation}
In order to analyze this probability in the following two lemmas we show that the order statistics $U^t_{\lceil (1-\bar\delta)n_t\rceil}$ and $\bar U^t_{|\L_{t+1}|-\lceil (1-\bar\delta)n_t\rceil+1}$ are of central orders and specify their distributions.
\begin{lemma}\label{lmm:U1}
The order statistic $U^t_{\lceil (1-\bar\delta)n_t\rceil}$ is of central order and is almost surely distributed as ${\cal N}(M_t ,\sigma_t^2)$ for
\begin{equation}\label{eq:central_distribution1}
M_t\dff\mu_0t+\sqrt{2t}\;{\rm erf}^{-1}(2\bar\delta-1)\;\quad\mbox{and}\quad \sigma^2_t\dff \frac{2\pi t\;\bar\delta(1-\bar\delta)}{{n}_t\exp(2{\rm erf}^{-1}(2\bar\delta-1))}\ .
\end{equation}
\end{lemma}
\begin{proof}
Note that in the asymptote of large $n$, we almost surely have $n_t=\omega(1)$ and Definition \ref{def:central} confirms that $U^t_{\lceil (1-\bar\delta)n_t\rceil}$ is an order statistic of {\em central} order from a sequence of $\bar n_t$ i.i.d. random variables with the parent distribution ${\cal N}\left(\mu_1t,t\right)$. By using Theorem~\ref{th:central} in the asymptote of large $n$ we have $U^t_{\lceil (1-\bar\delta)n_t\rceil}\sim{\cal N}(M_t,\sigma^2_t)$ for $M_t$ and $\sigma^2_t$ defined in \eqref{eq:central_distribution1}\footnote{For computing $M_t$ we have used the property that when $G$ denotes the cdf of the distribution ${\cal N}(\mu,\sigma^2)$, we have $G^{-1}(\alpha)=\mu+\sigma\sqrt{2}\;{\rm erf}(2 \alpha-1)$ for $\alpha\in(0,1)$.}.
\end{proof}
\begin{lemma}\label{lmm:U0}
The order statistic $\bar U^t_{|\L_{t+1}|-\lceil (1-\bar\delta)n_t\rceil+1}$ is of central order and is almost surely distributed as ${\cal N}(\bar M_t ,\bar\sigma_t^2)$ for
\begin{equation}\label{eq:central_distribution}
\bar M_t\dff\mu_1t+\sqrt{2t}\;{\rm erf}^{-1}(2\alpha-1)\;,\quad\mbox{and}\quad\bar\sigma^2_t\dff\frac{2\pi t\;\alpha(1-\alpha)}{\bar{n}_t\exp(2{\rm erf}^{-1}(2\alpha-1))}\ ,
\end{equation}
where ${\rm erf}(\cdot)$ denotes the Gauss error function.
\end{lemma}
\begin{proof}
The proofs follows the same line of arguments as in the proof of Lemma~\ref{lmm:U1} and noting that in the asymptote of large $n$
\begin{equation*}
    \frac{|\L_{t+1}|-\lceil (1-\bar\delta)n_t\rceil+1}{\bar n_{t}} \; = \;  \frac{\lfloor\alpha(n_t+\bar n_t-T_n)\rfloor-\lceil (1-\bar\delta)n_t\rceil+1}{\bar n_{t}}\; \xrightarrow{n\rightarrow\infty}\; \alpha
\end{equation*}
holds almost surely.
\end{proof}
Next, by defining
\begin{equation}\label{eq:sigma}
    \hat \sigma^2_t\dff \sigma_t^2+\bar\sigma^2_t\ ,
\end{equation}
and by using \eqref{eq:central_distribution}-\eqref{eq:central_distribution1} we obtain after a refinement action at time $t$ (i.e., $\psi(t)=1)$
\begin{eqnarray} \label{eq:n0_evolution2}
\nonumber \P\left(\lceil(1-\bar\delta)n_{t}\rceil\;\leq\;  n_{t+1}\med n_t\right) \overset{\eqref{eq:retain_event}}{=}  1-\int_{\frac{M_t-\bar M_t}{\hat\sigma_t}}^{\infty}\frac{\exp(-t^2/2)}{\sqrt{2\pi}}\;dt> 1-\frac{M_t-\bar M_t}{\hat\sigma_t}\cdot\exp\left\{-\frac{1}{2}\cdot\left(\frac{M_t-\bar M_t}{\hat\sigma_t}\right)^2\right\}\ ,
\end{eqnarray}
where we have used $\int_{x}^\infty\frac{1}{\sqrt{2\pi}}\;\exp(-t^2/2)\;dt<\frac{1}{x}\cdot \exp(-x^2/2)$ for $x>0$.
Hence, for the number of rare events discarded by the refinement actions, collectively, we have
\begin{eqnarray}\label{eq:n_t3}
\nonumber \P\left(n_\tau\;\geq\; n_1(1-\bar\delta)^K\med n_1\right) & \geq &
\P\left(n_\tau\;\geq\; n_1\prod_{t:\psi(t)=1}(1-\bar\delta)\med n_1\right)\\
\nonumber\\
\nonumber &\geq &
\prod_{t:\psi(t)=1} \P\left(n_{t+1}\;\geq\; n_t(1-\bar\delta)\med n_t\right) \\
\nonumber\\
&\geq &
\nonumber \prod_{t:\psi(t)=1} \left[1-\frac{M_t-\bar M_t}{\hat\sigma_t}\cdot\exp\left\{-\frac{1}{2}\cdot\left(\frac{M_t-\bar M_t}{\hat\sigma_t}\right)^2\right\}\right] \\
\nonumber\\
& \geq & \;\;\left[1-\frac{M_t-\bar M_t}{\hat\sigma_1}\cdot\exp\left\{-\frac{1}{2}\cdot\left(\frac{M_t-\bar M_t}{\hat\sigma_1}\right)^2\right\}\right]^K \ ,
\end{eqnarray}
where the first inequality holds since $\sum_{t:\psi(t)=1}\leq K$ and the last inequality holds as $\hat\sigma_t$ is increasing in $t$. By setting $\delta\dff 1-(1-\bar\delta)^{1/K}$, Eq.~\eqref{eq:n_t3} indicates that a {\em sufficient} condition that ensures $n_\tau\geq (1-\delta)n_1$ almost surely is that
\begin{equation}\label{eq:mu_scaling_ratio}
    \frac{M_t-\bar M_t}{\hat\sigma_1}\xrightarrow{n\rightarrow\infty} 1\ .
\end{equation}
Now, from \eqref{eq:sigma} recall that $\hat \sigma^2_1$ has two summands, where one scales with $1/{\bar n_1}$ and the other one scales with $1/{n_1}$. By noting that the stopping time can be bounded by a constant (Lemma~\ref{lmm:tau}) it can be readily verified that \eqref{eq:mu_scaling_ratio} holds if
\begin{eqnarray*}
  \mu_0-\mu_1 \;=\; \omega\left(\sqrt{\frac{1}{\bar n_1}+\frac{1}{n_1}}\right) \ ,
\end{eqnarray*}
which in turn holds {\em almost surely} if $\mu_0-\mu_1 \;=\; \omega\left(n^{-{\varepsilon_n}/2}\right)$.

\section{Proof of Lemma \ref{lmm:detection}}
\label{app:lmm:detection}

By using \eqref{eq:Pn2} we have the connection between the error probability $\PP_n(\tau,\bar\psi(\tau))$ for given $\tau$ and $\bar\psi(\tau)$ and low-order statistics of two sequences of sizes $n_{\tau}$ and $\bar n_{\tau}$ with Gaussian parent distributions with different means. According to \eqref{eq:Pn2} we have
\begin{equation*}
\quad \PP_n(\tau,\bar\psi(\tau))=\P\Big(U^{\tau}_{T_n}> \bar U^{\tau}_{1}\Big)\ .
\end{equation*}
In order to find the asymptotic distribution and the associated minimal domain of attraction of $U^{\tau}_{T_n}$ we use the result of Lemma~\ref{lmm:gaussian_min}. For this purpose let us define
$h:\{x\in\mathbb{R}\;:\;x>1\}\rightarrow\mathbb{R}^+$ as
\begin{equation*}
        h(x)=2\ln x-\ln\ln x\ .
\end{equation*}
Based on the definition and distribution of $Z_{\tau}^i$ for $i\in\L_{\tau}$ given in \eqref{eq:Z}-\eqref{eq:Z_dist} let us define
\begin{eqnarray}
 \label{eq:W0} \forall i\in\L^0_{\tau}:\quad \bar W_i &\dff& h(\bar{n}_{\tau}) +\sqrt{h(\bar{n}_{\tau})}\;\underset{\sim{\cal N}(0,1)}{\underbrace{\left(\frac{Z_{\tau}^i-\mu_1\cdot\tau}{\sqrt{\tau}}\right)}}\ ,\\
 \label{eq:W1} \mbox{and}\quad \forall i\in\L^1_{\tau}:\quad W_i &\dff& h(n_{\tau}) +\sqrt{h(n_{\tau})}\;\underset{\sim{\cal N}(0,1)}{\underbrace{\left(\frac{Z_{\tau}^i-\mu_0\cdot\tau}{\sqrt{\tau}}\right)}}\ .
\end{eqnarray}
Given the definitions in \eqref{eq:W0}-\eqref{eq:W1} we obtain
\begin{eqnarray}
  \nonumber
  \PP_n(\tau,\bar\psi(\tau)) &=&  \P\Big(U^{\tau}_{T_n}>\bar U^{\tau}_{1}\Big)\\
  \nonumber&&\\
  \label{eq:star} & = & \P\Bigg(\frac{U^{\tau}_{T_n}-\mu_1\cdot\tau}{\sqrt{\tau}}> \frac{\bar U^{\tau}_{1}-\mu_0\cdot\tau}{\sqrt{\tau}}+\sqrt{\tau}\;(\mu_0-\mu_1)\Bigg)\\
  \nonumber&&\\
  \label{eq:Pn2_2} & = & \P\Bigg(W_{T_n:n_{\tau}} >A_n\cdot \bar W_{1:\bar{n}_{\tau}} +B_n\Bigg)\ ,
\end{eqnarray}
where we have defined
\begin{equation}\label{eq:A}
A_n\dff\sqrt{\frac{h(n_{\tau})}{h(\bar{n}_{\tau})}}\qquad\mbox{and}\qquad
B_n\dff \sqrt{h(n_{\tau})}\;\Big[\sqrt{h(n_{\tau})}\;-\sqrt{h(\bar{n}_{\tau})}\;+\sqrt{\tau}\;(\mu_0-\mu_1)\ \Big]\ .
\end{equation}
In order to assess $\PP_n(\tau,\bar\psi(\tau))$ given in \eqref{eq:Pn2_2} we find the distributions of $W_{T_n:n_{\tau}}$ and $\bar W_{1:\bar{n}_{\tau}}$. For this purpose, according to Lemma~\ref{lmm:gaussian_min} the cdfs of $W_{1:n_{\tau}}$ and $\bar W_{1:\bar{n}_{\tau}}$, denoted by $Q_{1:n_{\tau}}(\cdot;n_{\tau})$ and $\bar Q_{1:\bar{n}_{\tau}}(\cdot;\bar{n}_{\tau})$, respectively, satisfy
\begin{eqnarray*}
  \lim_{\bar{n}_{\tau}\rightarrow\infty} \bar Q_{1:\bar{n}_{\tau}}(w;\bar{n}_{\tau})&=& 1-\exp(-\exp(w-\ln 2\sqrt{\pi}))\ , \\
  \mbox{and}\quad  \lim_{n_{\tau}\rightarrow\infty} Q_{1:n_{\tau}}(w;n_{\tau})&=& 1-\exp(-\exp(w-\ln 2\sqrt{\pi}))\ .
\end{eqnarray*}
Furthermore, given the asymptotic distributions of the first order statistics by using Theorem~\ref{th:low} we can find the distributions of the $r^{th}$-other (low-order) statistics. By denoting the cdfs of $W_{r:n_{\tau}}$ and $\bar W_{r:\bar{n}_{\tau}}$ by $Q_{r:n_{\tau}}(w;n_{\tau})$ and $\bar Q_{r:\bar{n}_{\tau}}(w;\bar{n}_{\tau})$, respectively, we have
\begin{eqnarray}
  \label{eq:W0_cdf}\lim_{\bar{n}_{\tau}\rightarrow\infty} \bar Q_{r:\bar{n}_{\tau}}(w+\varsigma;\bar{n}_{\tau})&=& 1-\sum_{i=0}^{r-1}\frac{\exp(iw-\exp(w))}{i!}\ ,\quad\forall w\in\mathbb{R}\ , \\
  \label{eq:W1_cdf}\mbox{and}\qquad \lim_{n_{\tau}\rightarrow\infty} Q_{r:n_{\tau}}(w+\varsigma;n_{\tau})&=& 1-\sum_{i=0}^{r-1}\frac{\exp(iw-\exp(w))}{i!}\ ,\quad\forall w\in\mathbb{R} \ .
\end{eqnarray}
where
\begin{equation*}
    \varsigma\;\dff\;\ln 2\sqrt{\pi}\ .
\end{equation*}
The associated pdfs, consequently, are given by
\begin{eqnarray}
  \label{eq:W0_pdf}\lim_{\bar{n}_{\tau}\rightarrow\infty} \bar q_{r:\bar{n}_{\tau}}(w+\varsigma;\bar{n}_{\tau})&=& \sum_{i=0}^{r-1}\frac{(\exp(w)-i)\exp(iw-\exp(w))}{i!}\ ,\quad\forall w\in\mathbb{R}\ , \\
  \label{eq:W1_pdf}\mbox{and}\qquad \lim_{n_{\tau}\rightarrow\infty} q_{r:n_{\tau}}(w+\varsigma;n_{\tau})&=& \sum_{i=0}^{r-1}\frac{(\exp(w)-i)\exp(iw-\exp(w))}{i!}\ ,\quad\forall w\in\mathbb{R}\ .
\end{eqnarray}
By using the distributions given in \eqref{eq:W0_cdf}-\eqref{eq:W1_pdf} for given $\tau$ and $\bar\psi(\tau)$ we find lower and upper bounds on $\PP_n(\tau,\bar\psi(\tau))$. These bounds, in turn, serve as the basis for finding necessary and sufficient conditions on the scaling laws of $(\mu_1-\mu_1)$ that guarantee $\PP_n(\tau,\bar\psi(\tau))\xrightarrow{\tiny n\rightarrow\infty} 0$.
\begin{description}
  \item[Lower bound:] By replacing $W_{T_n:n_\tau}$ by $W_{1:n_\tau}$, based on~\eqref{eq:Pn2_2}, $\PP_n(\tau,\bar\psi(\tau))$ can be lower bounded for sufficiently large $n$ as follows:
\begin{eqnarray}
  \nonumber \PP_n(\tau,\bar\psi(\tau)) & \geq & \P\Bigg(W_{1:n_{\tau}} >A_n\cdot \bar W_{1:\bar{n}_{\tau}} +B_n\Bigg)\\
  \nonumber&&\\
  \nonumber&= & 1-\;\int_{-\infty}^{0}Q_{1;n_{\tau}}(\underset{\leq\; B_n}{\underbrace{A_nw+B_n}})\;\bar q_{1;\bar{n}_{\tau}}(w)\;dw-\int^{\infty}_{0}Q_{1;n_{\tau}}\underset{\leq \;\omega+B_n}{\underbrace{A_nw+B_n}}\;\bar q_{1;\bar{n}_{\tau}}(w)\;dw\\
  \nonumber&&\\
  \label{eq:Pn3}& \geq & 1-\;\int_{-\infty}^{0}Q_{1;n_{\tau}}(B_n)\;\bar q_{1;\bar{n}_{\tau}}(w)\;dw-\int^{\infty}_{0}Q_{1;n_{\tau}}(w+B_n)\;\bar q_{1;\bar{n}_{\tau}}(w)\;dw\\
  \nonumber&&\\
  \nonumber & = & 1-Q_{1;n_{\tau}}(B_n)\;\bar Q_{1;\bar{n}_{\tau}}(0)+\int^{\infty}_{0}(1-Q_{1;n_{\tau}}(w+B_n))\;\bar q_{1;\bar{n}_{\tau}}(w)\;dw-(1-\bar Q_{1;\bar{n}_{\tau}}(0))\\
  \nonumber&&\\
  \label{eq:Pn4} & = &\exp(-(\exp(B_n-\varsigma)))\left\{1-\exp(-\exp(-\varsigma))+\frac{\exp(-\exp(-\varsigma))}{\exp(B_n)+1}\right\}\ ,
\end{eqnarray}
This strictly positive lower bound on $\PP_n(\tau,\bar\psi(\tau))$ approaches zero if and only if $B_n=\omega(1)$.
\item[Upper bound:] From \eqref{eq:Pn2} for sufficiently large $n$ we have
\begin{eqnarray}
  \nonumber\PP_n(\tau,\bar\psi(\tau)) & = & \P\Bigg(W^{0}_{T_n:n_{\tau}} >A_n\cdot W^{1}_{1:\bar{n}_{\tau}} +B_n\Bigg)\\
  \nonumber&&\\
  \nonumber&= & 1-\;\int_{-\infty}^{0}Q^0_{T_n;n_{\tau}}(A_nw+B_n)\;q^1_{1;\bar{n}_{\tau}}(w)\;dw-\int^{\infty}_{0}Q^0_{T_n;n_{\tau}}(A_nw+B_n)\;q^1_{1;\bar{n}_{\tau}}(w)\;dw\\
  \nonumber&&\\
  \nonumber & \leq & 1-\;\int_{-\infty}^{0}Q^0_{T_n;n_{\tau}}(w+B_n)\;q^1_{1;\bar{n}_{\tau}}(w)\;dw-\int^{\infty}_{0}Q^0_{T_n;n_{\tau}}(B_n)\;q^1_{1;\bar{n}_{\tau}}(w)\;dw\\
  \nonumber&&\\
  \nonumber & = & \big[1-Q^0_{T_n;n_{\tau}}(B_n)\big]\big[1-Q^1_{1;\bar{n}_{\tau}}(0)\big]+\int_{-\infty}^{-\varsigma}\exp(w-\exp(w+B_n)-\exp(w))\;dw\ ,
\end{eqnarray}
which can be easily shown to approach zero if and only if $B_n=\omega(1)$.
\end{description}
Therefore, a necessary and sufficient condition for having both the lower and upper bounds $\PP_n(\tau,\bar\psi(\tau))$ approach zero is that $B_n=\omega(1)$, which subsequently becomes a necessary and sufficient condition for $\PP_n(\tau,\bar\psi(\tau))\xrightarrow{\tiny n\rightarrow\infty} 0$. According to \eqref{eq:A},  $B_n=\omega(1)$ is equivalent to
\begin{eqnarray}\label{eq:scaling_mean1}
\nonumber \frac{B_n}{h(n_\tau)\sqrt{2\ln n}}=\omega\left(\frac{1}{h(n_\tau)\sqrt{2\ln n}}\right) &  \overset{\eqref{eq:A}}{\equiv} & \sqrt{\tau}\;\cdot\frac{\mu_0-\mu_1}{\sqrt{2\ln n}} - \left(\frac{\sqrt{h({\bar n_\tau})}}{\sqrt{2\ln n}}-\frac{\sqrt{h(n_\tau)}}{\sqrt{2\ln n}}\right)=\omega\left(\frac{1}{h(n_\tau)\sqrt{2\ln n}}\right)\ ,
\end{eqnarray}
where $h(\cdot)$ is defined in \eqref{eq:h} and it holds in the asymptote of large $n$ almost surely if and only if
\begin{equation*}
\sqrt{\tau}\;\cdot \sqrt{r_{\rm m}}- \left(1-\sqrt{{\varepsilon_n}}\right)=\omega\left(\frac{1}{h(n_\tau)\sqrt{2\ln n}}\right)\ ,
\end{equation*}
or equivalently
\begin{equation*}
  r_{\rm m}\;>\;\frac{(1-\sqrt{{\varepsilon_n}})^2}{\tau}\ .
\end{equation*}

\section{Proof of Lemma \ref{lmm:tau}}
\label{app:lmm:tau}

From the definition of $\psi(t)$ given in \eqref{eq:psi} we have
\begin{eqnarray}\label{eq:tau_card1}
\nonumber \forall t\in\{1,\dots,\tau-1\}:\quad |\L_{t+1}| & = & |\L_t|\cdot \mathds{1}_{\psi(t)=0}+\left(\lfloor\alpha(|\L_t|-T_n)\rfloor +T_n\right)\cdot \mathds{1}_{\psi(t)=1}\\
 & \geq & |\L_t|\cdot \mathds{1}_{\psi(t)=0}+\left(\alpha|\L_t|-1\right)\cdot \mathds{1}_{\psi(t)=1}\ ,
\end{eqnarray}
where $\mathds{1}_{\{\cdot\}}$ is the indicator function. Following this relationship, the number of sequences retained and fed to the detector is related to the initial number of sequences $n$ according to
\begin{equation}\label{eq:tau_card2}
\forall t\in\{1,\dots,\tau-1\}:\quad |\L_t|\;\geq\; \alpha^{\sum_{u=1}^{t}\psi(t)}\cdot|\L_1|-\sum_{u=0}^{t}\alpha^t\;\overset{\eqref{eq:P}}{\geq} \; \alpha^K|\L_1|-\frac{1-\alpha^K}{1-\alpha}\;=\; \alpha^K\cdot n-\frac{1-\alpha^K}{1-\alpha}\ .
\end{equation}
The constraint $\frac{1}{n}\sum_{t=1}^\tau|\L_t|\leq S$ in conjunction with \eqref{eq:tau_card2} provides
\begin{equation}
S\cdot n\;\geq\; \sum_{t=1}^\tau|\L_t| \geq \tau\cdot \left(\alpha^K\cdot n-\frac{1-\alpha^K}{1-\alpha}\right)\ ,
\end{equation}
which shows that in the asymptote of large $n$ we have $\tau\leq S/\alpha^K$.

\section{Proof of Lemma \ref{lmm:chi_min}}
\label{app:lmm:chi_min}

Let us define $Y_1,\dots, Y_m$ as i.i.d. random variables distributed according to $\chi^2(k)$. For the cdf of $\chi^2(k)$ random variables we have
\begin{equation}\label{eq:chi_cdf}
    F_Y(y)\;=\;\P(Y_i\;\leq\; y)\;=\;\frac{1}{2^{k/2}\Gamma(k/2)}\int_0^yt^{k/2-1}\exp(-t/2)\;dt\ .
\end{equation}
As for $t\in[0,y]$ we have $\exp(-y/2)\leq \exp(-t/2)\leq 1$ we obtain the following bounds on $F_Y(y)$:
\begin{equation}\label{eq:chi_cdf_bound}
    \frac{1}{2^{k/2}\Gamma(k/2)}\cdot \exp(-y/2)\int_0^yt^{k/2-1}\;dt \;\leq\; F_Y(y)\;\leq \frac{1}{2^{k/2}\Gamma(k/2)}\int_0^yt^{k/2-1}\;dt\ ,
\end{equation}
which subsequently provides
\begin{equation}\label{eq:chi_cdf_bound2}
    \frac{1}{2^{k/2}\Gamma(k/2+1)}\cdot\exp(-y/2)y^{k/2}\;dt \;\leq\; F_Y(y)\;\leq \frac{1}{2^{k/2}\Gamma(k/2)}\cdot y^{k/2}\ .
\end{equation}
Next, by setting
\begin{equation*}
 a_m=0\ ,\quad\mbox{and}\quad b_m=\frac{1}{2}\cdot\left[\frac{m}{\Gamma(k/2+1)}\right]^{2/k},\qquad\forall m\in\{2,3,\dots\}\ ,
\end{equation*}
and defining
\begin{equation*}
    W_i\;=\;a_m\;+\;b_mY_i
\end{equation*}
for the cdf of $W_i$, denoted by $F_W(w)$, we have $F_W(w)=\P(W_i\leq w) = F_Y(W_i\leq w/b_m)$. By invoking the inequalities in \eqref{eq:chi_cdf_bound2} we find the following bounds on $F_W(w)$.
\begin{equation}\label{eq:W_cdf_bound}
    \exp\left(-w\cdot \left[\frac{\Gamma(k/2+1)}{m}\right]^{2/k}   \right)\cdot\frac{w^{k/2}}{m} \;\leq\; F_W(w)\;\leq \frac{w^{k/2}}{m}\ ,
\end{equation}
which implies that in the asymptote of large $m$ we have $F_W(w)=\frac{w^{k/2}}{m}\;(1+o(1))$. Therefore, for the cdf of $W_{1:m}$ (the first order statistic of $\{W_1,\dots,W_m\}$, denoted by $Q_{1:m}(\cdot;m)$, in the asymptote of large $m$ we have
\begin{eqnarray*}
  Q_{1:m}(\cdot;m) &{=}&
  1-\left[1-F_W(w)\right]^m\\
  &=& 1-\left[1-\frac{w^{k/2}}{m}\;(1+o(1))\right]^m\\
  & = & 1-\exp\left(-w^{k/2}\right)\ .
\end{eqnarray*}

\section{Proof of Lemma \ref{lmm:chi_max}}
Let us define $Y_1,\dots, Y_m$ as i.i.d. random variables distributed according to $\chi^2(k)$ with cdf $F_Y(y)$ and define the function $g:\{x\in\mathbb{R}\;:\;x>1\}\rightarrow\mathbb{R}^+$
as
\begin{equation*}
    u(x)\dff -\left(\ln x+\frac{K-1}{2}\ln\ln x\right)\ .
\end{equation*}
By defining
\begin{eqnarray}\label{eq:Wi_last}
    \bar W_i & = & c_m\;+\;d_m\;Y_i\ ,
\end{eqnarray}
and setting
\begin{equation}\label{eq:cmdm}
    c_m\;=\;u(m)\ ,\qquad\mbox{and}\qquad b_m=\frac{1}{2}\ ,
\end{equation}
for the cdf of $\bar W_{m:m}$, denoted by $\bar Q_{m:m}(\cdot;m)$, we have
\begin{eqnarray}\label{eq:2_Q}
  \nonumber \bar Q_{m:m}(w;m) &{=}& \P\left(W_{m:m}\leq w\right)\\
  \nonumber  & \overset{\eqref{eq:Wi_last}}{=} & \left[\P\left(Y_{m:m}\leq \frac{w-c_m}{d_m}\right)\right]^m\\
  &=& \left[F_Y\left(\frac{w-c_m}{d_m}\right)\right]^m\ .
\end{eqnarray}
We next show that
\begin{equation}\label{eq:loglog2}
    \lim_{x\rightarrow\infty}x\cdot \log\big[F_Y\left(2(w-u(x))\right)\big]=-\frac{\exp(-w)}{\Gamma(k/2)}\ ,\qquad\forall w\in\mathbb{R}\ .
\end{equation}
By using L'H\^{o}pital's rule we obtain
\begin{align*}
 \lim_{x\rightarrow\infty}&\frac{\log\big[F_Y\left(2(w-u(x))\right)\big]}{\frac{1}{x}} \\
 &\\
&=  \lim_{x\rightarrow\infty}\frac{-2u'(x)\;f_Y(2(w-u(x)))}{-\frac{1}{x^2}}\\
&\\
& = \lim_{x\rightarrow\infty}\frac{1}{2^{k/2}\Gamma(k/2)}\cdot \frac{-\frac{2}{x}\cdot (1+\frac{k/2-1}{\log x})\cdot 2^{k/2-1}(w+\ln x+\frac{k-1}{2}\ln\ln x)^{k/2-1}\cdot\frac{1}{x}\cdot\left(\frac{1}{\ln x}\right)^{k/2-1}\cdot \exp(-w)}{-\frac{1}{x^2}}\\
&\\
& = - \frac{1}{\Gamma(k/2)}\cdot \exp(-w)\lim_{x\rightarrow\infty}\left(\frac{w+\ln x+\frac{k-1}{2}\ln\ln x}{\ln x}\right)^{k/2-1}\\
&\\
 &= - \frac{1}{\Gamma(k/2)}\cdot \exp(-w)\ .
\end{align*}
Equations \eqref{eq:2_Q} and \eqref{eq:loglog2} demonstrate that
\begin{eqnarray}
  \lim_{m\rightarrow\infty} \bar Q_{m:m}(w;m)=\exp\left(\frac{1}{\Gamma(k/2)}\cdot \exp(-w)\right)\ ,
\end{eqnarray}
which is the desired result.

\label{app:lmm:chi_max}

\section{Proof of Lemma \ref{lmm:refinement_var}}
\label{app:lmm:refinement_var}

From the definitions of of $n_t$ and $\bar n_t$ for all $t\in\{1,\dots,\tau \}$ we have $|\L_t|=n_t+\bar n_t$. Therefore, for any time $t\in\{1,\dots,\tau-1 \}$ that the refinement action is taken (i.e., $\psi(t)=1$) the condition
\begin{equation*}
    n_t \;=\; n_{t+1}
\end{equation*}
is equivalent to having
\begin{equation*}
    \bar n_{t+1} \;= \; |\L_{t+1}|-n_{t}\ .
\end{equation*}
By taking into account the distributions of the order statistics $\bar V^t_j$ and $V^t_j$ given in \eqref{eq:U0_var} and \eqref{eq:U1_var}, respectively, the event that the $|\L_{t+1}|$ retained events after a refinement action at time $t$ contains $n_t$ rare events is equivalent to
\begin{equation}\label{eq:retain_event}
     V^t_{n_t}\;<\; \bar V^t_{|\L_{t+1}|-n_t+1}\ .
\end{equation}
In order to analyze this probability in the following lemma we show that the order statistic $\bar V^t_{|\L_{t+1}|-n_t+1}$ is of central order and specify their distributions.
\begin{lemma}\label{lmm:U1_var}
The order statistic $\bar V^t_{|\L_{t+1}|-n_t+1}$ is of central order and
\begin{equation}\label{eq:central_distribution_var}
\bar V^t_{|\L_{t+1}|-n_t+1}\sim\;A_0\cdot\tilde A\ ,
\end{equation}
where
\begin{equation}\label{eq:central_distribution_var2}
    \tilde A \sim {\cal N}\left(\bar\mu\;,\;\bar\sigma^2\right)\ ,\quad \mbox{for}\quad \bar\mu\dff\bar G^{-1}(\alpha)\;,\quad\mbox{and}\quad\bar\sigma^2\dff\frac{\alpha(1-\alpha)}{\bar{n}_t|\bar g(\bar G^{-1}(\alpha))|^2}\ ,
\end{equation}
and $\bar G$ and $\bar g$ denote the cdf and pdf of $\chi^2(t)$.
 \end{lemma}
\begin{proof}
The proof of this follows from the same line of argument as in the proof of Lemma \ref{lmm:U1}.
\end{proof}
Therefore from \eqref{eq:central_distribution_var}-\eqref{eq:central_distribution_var2} we obtain when $\psi(t)=1$
\begin{eqnarray} \label{eq:n0_evolution2_var}
\P\left(n_{t+1}\;=\; n_t\med n_t\right) &=& \int_{-\infty}^{\infty}\P\left(\frac{V^t_{n_t}}{2A_1}-\ln n_t\leq \frac{A_0}{2A_1}\cdot a-\ln n_t\right)f_{\tilde A}(a)\;da\ ,
\end{eqnarray}
where $f_{\tilde A}$ denotes the pdf of $\tilde A$. Based on the definition of $V^t_j$ given in \eqref{eq:U0_var}, $V^t_{n_t}$ is the highest order statistic of a sequence of $n_t$ i.i.d. random variables with the parent distribution $\chi^2(t)$. Hence, by using Lemma~\ref{lmm:chi_max} we obtain
\begin{equation}\label{eq:n0_evolution2_var2}
\P\left(\frac{V^t_{n_t}}{2A_1}-\ln n_t\leq w\right)=\exp\left(\frac{1}{\Gamma(t/2)}\cdot\exp(-w)\right)\ .
\end{equation}
Therefore, from \eqref{eq:n0_evolution2_var} and \eqref{eq:n0_evolution2_var2} we get
\begin{eqnarray}
\nonumber \P\left(n_{t+1}\;=\; n_t\right) &=& \int_{-\infty}^{\infty}\exp\left(\frac{1}{\Gamma(t/2)}\cdot\exp\left(\ln n_t-\frac{A_0}{2A_1}\cdot a\right)\right)f_{\tilde A}(a)\;da\\
\nonumber &&\\
\label{eq:n0_evolution2_var3} &\geq & \int_{\bar\mu/2}^{3\bar\mu/2}\exp\left(\frac{1}{\Gamma(t/2)}\cdot\exp\left(\ln n_t-\frac{A_0}{2A_1}\cdot a\right)\right)f_{\tilde A}(a)\;da\\
\nonumber &&\\
\label{eq:n0_evolution2_var4} &\geq & \int_{\bar\mu/2}^{3\bar\mu/2}\exp\left(\frac{1}{\Gamma(t/2)}\cdot\exp\left(\ln n_t-\frac{A_0}{2A_1}\cdot \frac{\bar\mu}{2}\right)\right)f_{\tilde A}(a)\;da\\
\nonumber &&\\
\nonumber & = & \exp\left(\frac{1}{\Gamma(t/2)}\cdot\exp\left(\ln n_t-\frac{A_0}{2A_1}\cdot \frac{\bar\mu}{2}\right)\right)\;\int_{\bar\mu/2}^{3\bar\mu/2}f_{\tilde A}(a)\;da\\
\nonumber &&\\
\nonumber & = & \exp\left(\frac{1}{\Gamma(t/2)}\cdot\exp\left(\ln n_t-\frac{A_0}{2A_1}\cdot \frac{\bar\mu}{2}\right)\right)\;\P(|\tilde A-\bar\mu|\leq\bar\mu/2)\\
\nonumber &&\\
\label{eq:n0_evolution2_var5} & = & \exp\left(\frac{1}{\Gamma(t/2)}\cdot\exp\left(\ln n_t-\frac{A_0}{2A_1}\cdot \frac{\bar\mu}{2}\right)\right)\; \left[1-\left(\frac{2\bar\sigma}{\bar\mu}\right)^2\right]\ ,
\end{eqnarray}
where \eqref{eq:n0_evolution2_var3} holds by narrowing the interval of integral, \eqref{eq:n0_evolution2_var4} holds as the integrand in \eqref{eq:n0_evolution2_var3} takes its smallest value at $a=\bar\mu/2$, and \eqref{eq:n0_evolution2_var5} holds according to Chebyshev's  inequality. Next, note that in the asymptote of large $n$ almost surely we have $n_t\omega(1)$, and consequently, $\left(\frac{2\bar\sigma}{\bar\mu}\right)^2\xrightarrow{n\rightarrow\infty}0$. Hence,
\begin{equation}\label{eq:suff}
    \frac{A_0}{2A_1}\cdot\frac{\bar\mu}{2}-\ln n_t =\omega(1)
\end{equation}
is a sufficient condition for ensuring $\P\left(n_{t+1}\;=\; n_t\right)\xrightarrow{n\rightarrow\infty} 1$. In turn, this condition is satisfied in the asymptote of large $n$ almost surely when
$\frac{A_0}{A_1}=\omega(\ln n\epsilon_n)$ or equivalently $\frac{A_0}{A_1}=\omega({\varepsilon_n}\ln n)$.

\section{Proof of Lemma \ref{lmm:detection_var}}
\label{app:lmm:detection_var}

According to \eqref{eq:Pn2_var}, the detection error probability is
\begin{equation*}
\quad \PP_n(\tau,\bar\psi(\tau))=\P\Big(V^{\tau}_{T_n}\;>\; \bar V^{\tau}_{1}\Big)\ .
\end{equation*}
In order to find the asymptotic distribution and the associated minimal domains of attraction of $V^{\tau}_{T_n}$ and $\bar V^{\tau}_1$ we use the result of Lemma~\ref{lmm:chi_min}. For this purpose, let us define
\begin{equation}\label{eq:h_bar}
\bar h (x)=\frac{1}{2}\cdot x^{2/\tau}\ .
\end{equation}
Also, let us set
\begin{eqnarray}
 \label{eq:W0_var} \forall i\in\L^0_{\tau}:\quad \bar W_i &\dff& \bar Z_t^i\cdot\frac{\bar h(\bar n_{\tau})}{A_{0}}\ ,\\
 \label{eq:W1_var} \mbox{and}\quad \forall i\in\L^1_{\tau}:\quad W_i &\dff& \bar Z_t^i\cdot\frac{\bar h(n_{\tau})}{A_{1}}\ .
\end{eqnarray}
Therefore, according to Lemma~\ref{lmm:chi_min} the cdfs of $W_{1:n_{\tau}}$ and $\bar W_{1:\bar{n}_{\tau}}$, denoted by $Q_{1:n_{\tau}}(w;n_{\tau})$ and $\bar Q_{1:\bar{n}_{\tau}}(w;\bar{n}_{\tau})$, respectively, satisfy
\begin{eqnarray*}
  \lim_{\bar n_{\tau}\rightarrow\infty} \bar Q_{1:\bar n_{\tau}}(w;\bar n_{\tau})&=& 1-\exp\left(-w^{\tau/2}\right)\ , \\
  \mbox{and}\quad \lim_{{n}_{\tau}\rightarrow\infty} Q_{1:{n}_{\tau}}(w;{n}_{\tau})&=& 1-\exp\left(-w^{\tau/2}\right)\ .
\end{eqnarray*}
Furthermore, by using Theorem~\ref{th:low} we can find the distribution of the $T_n^{th}$-other (low-order) statistic $ W_{T_n: n_{\tau}}$ as follows:
\begin{eqnarray}
  \label{eq:W1_cdf_var}\lim_{\bar{n}_{\tau}\rightarrow\infty} Q_{T_n:{n}_{\tau}}(w)&=& 1-\exp\left(-w^{\tau/2}\right)\sum_{i=0}^{T_n-1}\frac{w^{i\tau/2}}{i!}\ ,\quad\forall w\in\mathbb{R}\ .
\end{eqnarray}
Given the definitions in \eqref{eq:W0_var}-\eqref{eq:W1_var} we obtain
\begin{eqnarray*}
  \PP_n(\tau,\bar\psi(\tau)) &=&  \P\Big(V^{\tau}_{T_n}>\bar V^{\tau}_{1}\Big)\\
 &=&  1-\P\left(W_{T_n:n_\tau}\cdot\frac{A_{1}}{\bar h( n_{\tau})} < \bar W_{1:\bar n_\tau}\cdot\frac{A_{0}}{\bar h(\bar n_{\tau})}\right) \\
 &=&  1-\P\left(W_{T_n:n_\tau} < \bar W_{1:\bar n_\tau}\cdot\frac{A_{0}}{A_{1}}\cdot\frac{\bar h( n_{\tau})}{\bar h(\bar n_{\tau})}\right) \ .
\end{eqnarray*}
By setting
\begin{equation}\label{eq:Theta}
\Theta=\frac{A_0}{A_1}\cdot\frac{\bar h( n_{\tau})}{\bar h(\bar n_{\tau})}\ ,
\end{equation}
we get
\begin{eqnarray}
\nonumber \PP_n(\tau,\bar\psi(\tau)) & = & 1- \int_0^\infty \bar q_{1:\bar n_{\tau}}(w;\bar n_{\tau}) \int_0^{\Theta} q_{T_n;n_{\tau}} (x;n_{\tau})\  dx \ dw\\
\nonumber &&\\
\nonumber & = & 1- \int_0^\infty \bar q_{1:\bar n_{\tau}}(w;\bar n_{\tau})\; Q_{T_n;n_{\tau}}\left(w\cdot\frac{A_0}{A_1}\cdot\Theta \ ;\ n_{\tau}\right) \ dw\\
\nonumber &&\\
\nonumber & \overset{\eqref{eq:W1_cdf_var}}{=} & 1- \int_0^\infty \underset{=\tau/2\;w^{\tau/2-1}\;e^{-w^{\tau/2}}}{\underbrace{\bar q_{1;\bar n_{\tau}}(w;\bar n_{\tau})}} \left(1-\exp\left(-w^{\tau/2}\; \Theta^{\tau/2}\;\right)\sum_{i=0}^{T_n-1}\ \frac{w^{i\tau/2}\;\Theta^{i\tau/2}}{i!}\right)\ dw\\
\nonumber &&\\
\nonumber &= & 1- \left\{1-\sum_{i=0}^{T_n-1}\frac{\tau/2\;\Theta^{i\tau/2}}{i!}\int_0^\infty \exp\left(-w^{\tau/2}\left(1+ \Theta^{\tau/2}\right)\right)\ w^{i\tau/2+\tau/2-1}\ dw\right\}\ .
\end{eqnarray}
By further setting
\begin{equation*}
    s\dff w^{\tau/2}\left(1+ \Theta^{\tau/2}\right)\ ,
\end{equation*}
we find
\begin{eqnarray}\label{eq:Pn_var2}
\nonumber \PP_n(\tau,\bar\psi(\tau))
& = & 1- \left\{1-\sum_{i=0}^{T_n-1}\frac{\tau/2\;\Theta^{i\tau/2}}{i!}\cdot\frac{1}{\tau/2\;\left(1+\Theta^{\tau/2}\right)^{i+1}}\underset{=\ \Gamma(i+1)\ =\ i!}{\underbrace{\int_0^\infty \exp(-s)\  s^i\ ds}}\right\}\\
\nonumber \nonumber &&\\
\nonumber & = & 1- \left\{1-\sum_{i=0}^{T_n-1}\frac{\Theta^{i\tau/2}}{\left(1+\Theta^{\tau/2}\right)^{i+1}}\right\}\\
\nonumber \nonumber &&\\
\nonumber & = & 1- \left\{1-\frac{1}{1+\Theta^{\tau/2}}\sum_{i=0}^{T_n-1}\left(\frac{\Theta^{\tau/2}}{1+\Theta^{\tau/2}}\right)^i\right\}\\
\nonumber \nonumber &&\\
& = & 1- \left(\frac{\Theta^{\tau/2}}{1+\Theta^{\tau/2}}\right)^{T_n}\ .
\end{eqnarray}
Hence, the requirement $\PP_n(\tau,\bar\psi(\tau)) \xrightarrow{n\rightarrow\infty} 0$ is equivalent to $\Theta^{\tau/2}=\omega(1)$, which by taking into account \eqref{eq:h_bar}~and~\eqref{eq:Theta}, is in turn equivalent to
\begin{equation}\label{eq:A_scaling}
\frac{A_{0}}{A_{1}}\;=\; \omega\left(\sqrt[\tau/2]{\frac{\bar n_{\tau}}{ n_{\tau}}}\right)\ .
\end{equation}
By taking into account that $\frac{\bar n_\tau}{n_\tau}\xrightarrow{n\rightarrow\infty} c\cdot n^{1-{\varepsilon_n}}$ for some constant $c\in\mathbb{R}_+$, this condition can be equivalently cast as
\begin{equation}\label{eq:A_scaling2}
\xi_{\rm v}\dff \frac{\ln\frac{A_{0}}{A_{1}}}{\ln n}\;>\;\frac{2(1-{\varepsilon_n})}{\tau}\ .
\end{equation}

\section{Proof of Theorem \ref{th:stop}}
\label{app:th:mean}
From the definition of $\sP_n(S,K)$, the optimal stopping time $\tau$ and the optimal switching sequence $\bar\psi(\tau)$ are the minimizers of $\PP_n(\tau,\bar\psi(\tau))$ within the constraints on the sampling budget and the number of switchings. In the sequel we first prove that the impacts of the stopping time and the switching sequence are embedded in the term $\tau$ and minimizing $\PP_n(\tau,\bar\psi(\tau))$ reduces to maximizing $\tau$. We provide separate proofs for the mean and variance cases.\vspace{.1 in}\\
  \underline{\bf Mean:}\\
As shown in \eqref{eq:Pn2} and \eqref{eq:star} we have
\begin{eqnarray}\label{eq:Pn5}
  \PP_n(\tau,\bar\psi(\tau)) &=&
  P\Bigg(\frac{U^{\tau}_{T_n}-\mu_1\cdot\tau}{\sqrt{\tau}}> \frac{\bar U^{\tau}_{1}-\mu_0\cdot\tau}{\sqrt{\tau}}+\sqrt{\tau}\;(\mu_0-\mu_1)\Bigg)\ ,
\end{eqnarray}
where $U^\tau_{T_n}$ and $\bar U^\tau_1$ denote the order statistics of the sets
$\{Z^i_t:\; t\in\L_t^1\}$ and $\{Z^i_t:\; t\in\L_t^0\}$. Furthermore, as shown in \eqref{eq:W0} and \eqref{eq:W1} we also have
\begin{eqnarray*}
&&\forall i\in\L^0_{\tau}:\quad  \frac{Z_{\tau}^i-\mu_0\cdot\tau}{\sqrt{\tau}}\;\sim\;{\cal N}(0,1)\ ,\\
\mbox{and}\quad &&\forall i\in\L^1_{\tau}:\quad  \frac{Z_{\tau}^i-\mu_1\cdot\tau}{\sqrt{\tau}}\;\sim\;{\cal N}(0,1)\ .
\end{eqnarray*}
Hence, the distributions of $\frac{U^\tau_{T_n}-\mu_1\cdot\tau}{\sqrt{\tau}}$ and $\frac{\bar U_{\tau}^1-\mu_0\cdot\tau}{\sqrt{\tau}}$ are independent of $\tau$ and the effect of $\tau$ on $ \PP_n(\tau,\bar\psi(\tau))$ is captured entirely by the term $\sqrt{\tau}(\mu_0-\mu_1)$. Therefore, the error probability is minimizes when $\tau$ is maximized. \\
\underline{\bf Variance:}\\
By recalling \eqref{eq:h_bar}, \eqref{eq:Theta}, and \eqref{eq:Pn_var2} we have
\begin{eqnarray*}
\PP_n(\tau,\bar\psi(\tau))
& = & 1- \left(\frac{\Theta^{\tau/2}}{1+\Theta^{\tau/2}}\right)^{T_n}\ ,
\end{eqnarray*}
where from \eqref{eq:Theta} and \eqref{eq:h_bar} we have
\begin{equation*}
    \Theta^{\tau/2}=\left(\frac{A_0}{A_1}\right)^{\tau/2}\cdot\frac{ n_{\tau}}{\bar n_{\tau}}\
\end{equation*}
Therefore, minimizing $\PP_n(\tau,\bar\psi(\tau))$ is equivalent to maximizing $\Theta^{\tau/2}$, which happens when $\tau$ is maximized. Next, we obtain the choices of $\tau$ and $\bar\psi(\tau)$ that maximize $\tau$. Let us denote the optimal number of refinement actions by $K^*\leq K$. We argue that the optimal switching sequence that maximizes $\tau$ is of the form
\begin{equation}\label{eq:psi_proof1}
\bar\psi(\tau)=\{\underset{K^*}{\underbrace{1,\dots , 1}},\;0,\dots, 0\}\ .
\end{equation}
To prove this we first show that the aggregate number of observations taken corresponding to the switching sequence
\begin{equation}\label{eq:hat_psi}
\hat\psi(\tau)\dff \{\hat\psi_1\;,\,0\;,\;1\;,\;\hat\psi_2\}
\end{equation}
is strictly smaller than that corresponding to the switching sequence
\begin{equation}\label{eq:hat_psi2}
\doublehat\psi(\tau)\dff \{\hat\psi_1\;,\,1\;,\;0\;,\;\hat\psi_2\}\ ,
\end{equation}
where  $\hat\psi_1$ and $\hat\psi_2$ are switching sub-sequences. Note that the sequences $\hat\psi(\tau)$ and $\doublehat\psi(\tau)$ differ only in two switch values. By denoting the lengths of $\hat\psi_1$ and $\hat\psi_2$ by $|\hat\psi_1|$ and $|\hat\psi_2|$, respectively, the number of samples taken according to the switching sequence $\hat\psi(\tau)$ is
\begin{equation}\label{eq:tau_1}
\sum_{t=1}^\tau|\L_t|\; =\; \underset{\mbox{\small times }t=1,\dots,|\hat\psi_1|+1}{\underbrace{\sum_{t=1}^{|\hat\psi_1|+1}|\L_t|}}\;+\;\underset{\mbox{\small time }t=|\hat\psi_1|+2}{\underbrace{|\L_{\hat\psi_1|+1}|}}\;+\;\underset{\mbox{\small time }t=|\hat\psi_1|+3}{\underbrace{\lfloor \alpha(|\L_{\hat\psi_1|+1}|-T_n)\rfloor +T_n}}+\sum_{t=|\hat\psi_1|+4}^{|\hat\psi_1|+|\hat\psi_2|+2}|\L_t|\ ,
\end{equation}
and the number of samples taken according to the switching sequence $\doublehat\psi(\tau)$ is
\begin{equation}\label{eq:tau_2}
\sum_{t=1}^\tau|\L_t|\; =\; \underset{\mbox{\small times }t=1,\dots,|\hat\psi_1|+1}{\underbrace{\sum_{t=1}^{|\hat\psi_1|+1}|\L_t|}}\;+\;\underset{\mbox{\small time }t=|\hat\psi_1|+2}{\underbrace{\lfloor \alpha(|\L_{\hat\psi_1|+1}|-T_n)\rfloor +T_n}}\;+\;\underset{\mbox{\small time }t=|\hat\psi_1|+3}{\underbrace{\lfloor \alpha(|\L_{\hat\psi_1|+1}|-T_n)\rfloor +T_n}}+\sum_{t=|\hat\psi_1|+4}^{|\hat\psi_1|+|\hat\psi_2|+2}|\L_t|\ .
\end{equation}
By comparing \eqref{eq:tau_1} and \eqref{eq:tau_2} all the summands of $\sum_{t=1}^\tau|\L_t|$ are identical except the second terms, where it is is strictly smaller in \eqref{eq:tau_1}. By following the same line of argument, it is concluded that among all switching sequences that contain $K^*$  switches with value 1, the sequence that has its initial $K^*$ switches set to 1, takes the least number of samples. Therefore, among all switching sequences with $K^*$ refinement actions, the sequence of the form \eqref{eq:psi_proof1} leaves the most unused sampling resources, which in turn can be exploited to take further samples and delay (increase) the stopping time.

Next, we need to determine the optimal number of refinement actions $K^*$. By using \eqref{eq:alpha_bar}, the following equation delineates the number of sequences of sequences observed at time $t$.
\begin{equation}\label{eq:t_observe}
|\L_t|\;=\; \left\lfloor \alpha^{\sum_{u=1}^{t-1}\psi(t)}(n-T_n)\right\rfloor+T_n\ .
\end{equation}
Hence, by invoking that the switching sequence is of form \eqref{eq:psi_proof1}, we obtain
\begin{equation}\label{eq:t_bound 4}
\sum_{t=1}^{\tau}|\L_t|\;=\; \sum_{t=1}^{K^*}\left(\lfloor\alpha^{t-1}( n-T_n)\rfloor+T_n\right)+(\tau-K^*)\left(\lfloor\alpha^{K^*}(n-T_n)\rfloor+T_n\right)\ .
\end{equation}
Furthermore, by taking into account the hard constraint on the sampling resources, i.e., $\sum_{t=1}^\tau|\L_t|\leq S\cdot n$, in the asymptote of large $n$ we have
\begin{equation}\label{eq:t_bound5}
\tau=\left\lfloor \frac{nS-\sum_{t=1}^{K^*}(\lfloor\alpha^{t-1}(n-T_n)+T_n\rfloor)}{\lfloor\alpha^{K^*}( n-T_n)\rfloor+T_n}\right\rfloor\;+\;K^*\; \doteq \;\left\lfloor S\cdot\alpha^{-K^*}+\frac{1-\alpha^{-K^*}}{1-\alpha}\right\rfloor \ .
\end{equation}
Also, by noting that $\alpha\in(0,1)$ it can be readily shown that for $x\in\{0,\mathbb{N}\}$
\begin{equation*}
S\cdot\alpha^{-x}+\frac{1-\alpha^{-x}}{1-\alpha}
\end{equation*}
is increasing in $x$ when $\alpha\leq 1-\frac{1}{S}$ and decreasing in $x$ when $\alpha>1-\frac{1}{S}$. As a result
\begin{equation*}
K^*=\left\{
\begin{array}{ll}
K\ , & \mbox{if}\;\;\alpha\leq 1-\frac{1}{S}\\
\;0\ , & \mbox{if}\;\;\alpha> 1-\frac{1}{S}\\
\end{array}\right.\ .
\end{equation*}
By using \eqref{eq:t_bound5} the optimal stopping time is
\begin{equation*}
\tau=\left\{
\begin{array}{ll}
\;K+ \left\lfloor S\cdot\alpha^{-K}+\frac{1-\alpha^{-K}}{1-\alpha}\right\rfloor \ , & \mbox{if}\;\;\alpha\leq 1-\frac{1}{S}\\
\;S\ , & \mbox{if}\;\;\alpha> 1-\frac{1}{S}\\
\end{array}\right.\ .
\end{equation*}

\section{Proof of Corollary \ref{cor:agility_mean}}
\label{app:cor:mean:agility}

By following the same line of argument as in the proof of Theorem~\ref{th:stop} it can be readily shown that the impact of the average sampling budget on the asymptotic error probability $\sP_n(S,K)$ is captured by $\tau$. Therefore, achieving identical asymptotic performance $\sP_n(S,K)\doteq \sP_n(S_0,0)$ is equivalent to equating the term $\tau$ under the adaptive and the non-adaptive procedures. By recalling the results of Theorem~\ref{th:stop} and \eqref{eq:switch}-\eqref{eq:t_optimal} this latter equivalence can be stated as
\begin{equation}\label{eq:agility_1}
    \left\lfloor S\cdot\alpha^{-K}+\frac{1-\alpha^{-K}}{1-\alpha}\right\rfloor=S_0\ ,
\end{equation}
which in turn provides
\begin{equation}\label{eq:agility_2}
    S\cdot\alpha^{-K}+\frac{\alpha-\alpha^{-K}}{1-\alpha}\;<\;S_0\;<\; S\cdot\alpha^{-K}+\frac{\alpha-\alpha^{-K}}{1-\alpha}\ .
\end{equation}
After some simple manipulations we obtain the desired result.

{\small
\bibliographystyle{IEEETran}
\bibliography{IEEEabrv,Search}

\begin{thebibliography}{10}
\providecommand{\url}[1]{#1}
\csname url@samestyle\endcsname
\providecommand{\newblock}{\relax}
\providecommand{\bibinfo}[2]{#2}
\providecommand{\BIBentrySTDinterwordspacing}{\spaceskip=0pt\relax}
\providecommand{\BIBentryALTinterwordstretchfactor}{4}
\providecommand{\BIBentryALTinterwordspacing}{\spaceskip=\fontdimen2\font plus
\BIBentryALTinterwordstretchfactor\fontdimen3\font minus
  \fontdimen4\font\relax}
\providecommand{\BIBforeignlanguage}[2]{{%
\expandafter\ifx\csname l@#1\endcsname\relax
\typeout{** WARNING: IEEEtran.bst: No hyphenation pattern has been}%
\typeout{** loaded for the language `#1'. Using the pattern for}%
\typeout{** the default language instead.}%
\else
\language=\csname l@#1\endcsname
\fi
#2}}
\providecommand{\BIBdecl}{\relax}
\BIBdecl

\bibitem{finance1}
E.~Andreou and E.~Ghysels, ``The impact of sampling frequency and volatility
  estimators on change-point tests,'' \emph{Journal of Financial Econometrics},
  vol.~2, no.~2, pp. 290--318, 2004.

\bibitem{finance2}
M.~Beibel and H.~R. Lerche, ``A new look at optimal stopping problems related
  to mathematical finance,'' \emph{Statistica Sinica}, vol.~7, pp. 63--108,
  1997.

\bibitem{finance3}
A.~N. Shiryaev, ``Quickest detection problems in the technical analysis of
  financial data,'' \emph{Mathematical Finance}, pp. 487--521, 2002.

\bibitem{telecom1}
H.~Li, C.~Li, and H.~Dai, ``Quickest spectrum sensing in cognitive radio,'' in
  \emph{Proc. 42nd Annu. Conf. Info. Sci. and Syst. (CISS)}, Princeton, NJ,
  Mar. 2008, pp. 203--208.

\bibitem{telecom2}
H.~Jiang, L.~Lai, R.~Fan, and H.~V. Poor, ``Cognitive radio: How to maximally
  utilize spectrum opportunities in sequential sensing,'' in \emph{Proc. IEEE
  Global Commun. Conf. (Globecom)}, New Orleans, LA, Dec. 2008.

\bibitem{econ1}
E.~Andersson, D.~Bock, and M.~Fris\'{e}n, ``Detection of turning points in
  business cycles,'' \emph{Journal of Business Cycle Measurement and Analysis},
  vol.~1, no.~1, pp. 93--108, 2004.

\bibitem{econ2}
------, ``Some statistical aspects of methods for detection of turning points
  in business cycles,'' \emph{Journal of Applied Statistics}, vol.~33, no.~3,
  pp. 257--278, 2006.

\bibitem{econ3}
E.~Andreou and E.~Ghysels, ``The impact of sampling frequency and volatility
  estimators on change-point tests,'' \emph{Journal of Financial Econometrics},
  vol.~2, no.~2, pp. 290--318, 2004.

\bibitem{econ4}
D.~W.~K. Andrews, I.~Lee, and W.~Ploberger, ``Optimal changepoint tests for
  normal linear regression,'' \emph{Journal of Econometrics}, vol.~70, no.~1,
  pp. 9--38, 1996.

\bibitem{econ5}
I.~Berkes, E.~Gombay, L.~Horv\'{a}th, and P.~Kokoszka, ``Sequential
  change-point detection in garch($p,q$) models,'' \emph{Econometric Theory},
  vol.~20, no.~6, pp. 1140--1167, 2004.

\bibitem{econ6}
L.~D. Broemling and H.~Tsurumi, \emph{Econometrics and Structural
  Change}.\hskip 1em plus 0.5em minus 0.4em\relax New York: Marcel Dekker,
  1987.

\bibitem{energy1}
S.~M. Amin and B.~F. Wollenberg, ``Toward a smart grid: {p}ower delivery for
  the 21st century.'' \emph{IEEE Power and Energy Mag.}, vol.~3, no.~5, pp.
  34--41, Sept.-Oct. 2005.

\bibitem{energy2}
X.~H. M.-O. Pun, C.-C. Kuo, and Y.~Zhao, ``A change-point detection approach to
  power quality monitoring in smart grids,'' in \emph{Proc. IEEE Int. Conf. on
  Commun. (ICC)}, Cape Town, Saouth Africa, May 2010.

\bibitem{sec1}
A.~A. Cardenas, J.~S. Baras, and V.~Ramezani, ``Distributed change detection
  for worms, {D}{D}o{S} and other network attacks,'' \emph{Proceedings of the
  2004 American Control Conference}, vol.~2, pp. 1008--1013, Jun. 2004.

\bibitem{sec2}
R.~K.~C. Chang, ``Defending against flooding-based distributed
  denial-of-service attacks: A tutorial,'' \emph{IEEE Communications Magazine},
  vol.~40, no.~10, pp. 42--51, 2002.

\bibitem{sec3}
H.~Kim, B.~L. Rozovskii, and A.~G. Tartakovsky, ``A nonparametric multichart
  {CUSUM} test for rapid detection of dos attacks in computer networks,''
  \emph{International Journal of Computing \& Information Sciences}, vol.~2,
  no.~3, pp. 149--158, 2004.

\bibitem{sec4}
P.~Papantoni-Kazakos, ``Algorithms for monitoring changes in the quality of
  communication links,'' \emph{IEEE Transactions on Communications}, vol.~27,
  no.~4, pp. 682--693, 1976.

\bibitem{sec5}
A.~G. Tartakovsky, B.~L. Rozovskii, R.~B. Blazek, and H.~Kim, ``A novel
  approach to detection of intrusions in computer networks via adaptive
  sequential and batch-sequential change-point detection methods,'' \emph{IEEE
  Transaction on Signal Processing}, vol.~54, no.~9, pp. 3372--3382, 2006.

\bibitem{sec6}
M.~Thottan and C.~Ji, ``Anomaly detection in ip networks,'' \emph{IEEE
  Transaction on Signal Processing}, vol.~51, no.~8, pp. 2191--2204, 2003.

\bibitem{fraud1}
M.~Fri\'{s}en, ``Statistical surveillance: Optimality and methods,''
  \emph{International Statistical Review}, vol.~71, no.~2, pp. 403--434, 2003.

\bibitem{fraud2}
------, ``Properties and use of the shewhart method and its followers,''
  \emph{Sequential Analysis}, vol.~26, no.~2, pp. 171--193, 2007.

\bibitem{sensing1}
C.~Han, P.~Willett, and D.~Abraham, ``Some methods to evaluate the performance
  of {P}age's test as used to detect transient signals,'' \emph{IEEE
  Transactions on Signal Processing}, vol.~47, no.~8, pp. 2112--2127, 1999.

\bibitem{sensing2}
M.~Marcus and P.~Swerling, ``Sequential detection in radar with multiple
  resolution elements,'' \emph{IEEE Transactions on Information Theory},
  vol.~8, no.~3, pp. 237--245, 1962.

\bibitem{sensing3}
V.~F. Pisarenko, A.~F. Kushnir, and I.~Savin, ``Statistical adaptive algorithms
  for estimation of onset moments of seismic plates,'' \emph{Physics of the
  Earth and Planetary Interiors}, vol.~47, pp. 4--10, 1987.

\bibitem{Wald:AMS49}
A.~Wald and J.~Wolfowitz, ``Character of the sequential probability ratio
  test,'' \emph{Annals of Mathematical Statistics}, vol.~19, no.~3, pp.
  326--339, Sep. 1948.

\bibitem{Lai:IT11}
L.~Lai, H.~V. Poor, Y.~Xin, and G.~Georgiadis, ``Quickest search over multiple
  sequences,'' \emph{IEEE Transactions on Information Theory}, vol.~57, no.~8,
  pp. 5375--5386, Aug. 2011.

\bibitem{Wald:AMS45}
A.~Wald, ``Sequential tests of statistical hypotheses,'' \emph{Annals of
  Mathematical Statistics}, vol.~16, no.~2, pp. 117--186, Jun. 1945.

\bibitem{Zigangirov:TPA66}
K.~S. Zigangirov, ``On a problem in optimal scanning,'' \emph{Theory Probab.
  Appl.}, vol.~11, no.~2, 1966.

\bibitem{Dragalin:Metrica96}
V.~Dragalin, ``A simple and effective scanning rule for a multi-channel
  system,'' \emph{Metrika}, vol.~43, no.~1, 1996.

\bibitem{Haupt:AISTATS09}
J.~Haupt, R.~Castro, and R.~Nowak, ``Distilled sensing: Selective sampling for
  sparse signal recovery,'' in \emph{Proc. 12th International Conference on
  Artificial Intelligence and Statistics (AISTATS)}, Clearwater Beach, FL, Apr.
  2009, pp. 216--223.

\bibitem{Haupt:IT11}
------, ``Distilled sensing: Adaptive sampling for sparse detection and
  estimation,'' \emph{IEEE Trans. Info. Theory}, vol.~57, no.~9, pp.
  6222--6235, Sep. 2011.

\bibitem{Malloy:ISIT11}
M.~Malloy and R.~Nowak, ``Sequential analysis in high-dimensional multiple
  testing and sparse recovery,'' in \emph{Proc. IEEE Int. Symp. Inf. Theory
  (ISIT)}, St. Petersburg, Russia, Jul.-Aug. 2011, pp. 2661--2665.

\bibitem{Reeves:ISIT08}
G.~Reeves and M.~Gastpar, ``Sampling bounds for sparse support recovery in the
  presence of noise,'' in \emph{Proc. IEEE Int. Symp. Inf. Theory (ISIT)},
  Toronto, Cadana, Jul. 2008, pp. 2187--2192.

\bibitem{Jin:2006_1}
D.~Donoho and J.~Jin, ``Asymptotic minimaxity of false discovery rate
  thresholding for sparse exponential data,'' \emph{Annals of Statistics},
  vol.~34, no.~6, pp. 2980--3018, 2006.

\bibitem{Fisher:Cambridge28}
R.~A. Fishera and L.~H.~C. Tippett, ``Limiting forms of the frequency
  distribution of the largest or smallest member of a sample,''
  \emph{Mathematical Proceedings of the Cambridge Philosophical Society},
  vol.~24, no.~2, pp. 180--190, 1928.

\bibitem{Galambos:book}
J.~Galambos, \emph{The Asymptotic Theory of Extreme Order Statistics},
  1st~ed.\hskip 1em plus 0.5em minus 0.4em\relax New York, NY: John Wiley \&
  Sons, 1978.

\bibitem{Arnold:Book}
B.~C. Arnold, N.~Balakrishnan, and H.~N. Nagaraja, \emph{A First course in
  order statistics}.\hskip 1em plus 0.5em minus 0.4em\relax New York, NY: John
  Wiley \& Sons, 1992.

\bibitem{Abramovich}
F.~Abramovich, Y.~Benjamini, D.~Donoho, and I.~Johnstone, ``Adapting to unknown
  sparsity by controlling the false discovery rate,'' \emph{Annals of
  Statistics}, vol.~34, no.~2, pp. 584--653, 2006.

\end{thebibliography}
}

\end{document}